\documentclass[10pt,twoside,twocolumn,journal]{IEEEtran}
\usepackage[utf8]{inputenc} 
\usepackage[T1]{fontenc}
\usepackage{url}
\usepackage{ifthen}
\usepackage{graphicx}
\usepackage[T1]{fontenc}

\usepackage{enumitem}
\usepackage[cmex10]{amsmath}
\usepackage{amssymb}
\usepackage{multicol}

\usepackage{cite}
\usepackage{amsthm}
\usepackage{textcomp}
\usepackage{siunitx}
\usepackage{color}
\usepackage{cases}
\usepackage[font={small}]{caption}
\usepackage{epstopdf}
\usepackage{graphicx}
\usepackage{caption}
\usepackage{verbatim} 
\usepackage{bbm}
\usepackage{mathtools}

\newtheorem{theorem}{Theorem}
\newtheorem{lemma}{Lemma}

\newtheorem{remark}{Remark}

\theoremstyle{definition}
\newtheorem{example}{Example}

\theoremstyle{remark}

\DeclareMathOperator{\argmax}{argmax} % thin space, limits on side in displays
\newtheorem*{theorem*}{Theorem}
\newtheorem{assumption}{Assumption}
\usepackage[ruled,vlined]{algorithm2e}
\setlist[description]{style=multiline}

\IEEEoverridecommandlockouts

\begin{document}
\sloppy

\title{Edge Computing in the Dark: Leveraging Contextual-Combinatorial Bandit and Coded Computing} 
\author{Chien-Sheng Yang,~\IEEEmembership{Student Member,~IEEE}, Ramtin Pedarsani,~\IEEEmembership{Member,~IEEE},\\ and A. Salman Avestimehr,~\IEEEmembership{Fellow,~IEEE}\vspace{-3mm}
        \thanks{%Manuscript received May 7, 2020; revised December 8, 2020; accepted January 17, 2021; approved by IEEE/ACM TRANSACTIONS ON NETWORKING Associate Editor G. Iosifidis. 
        This material is based upon work supported by Defense Advanced Research Projects Agency (DARPA) under Contract No. HR001117C0053, ARO award W911NF1810400, NSF grants CCF-1703575, CCF-1763673, CNS-2003035, CNS-2002874, ONR Award No. N00014-16-1-2189, UC Office of President under Grant LFR-18-548175 and a gift from Intel. The views, opinions, and/or findings expressed are those of the author(s) and should not be interpreted as representing the official views or policies of the Department of Defense or the U.S. Government. A preliminary part of this work was presented in IEEE ISIT 2020 \cite{yang2020online}.\textit{ (Corresponding author: Chien-Sheng Yang.)}}
        \thanks{C.-S.~Yang and A.~S.~Avestimehr are with the Department of Electrical and Computer Engineering, University of Southern California, Los Angeles, CA 90089 USA (e-mail: chienshy@usc.edu; avestimehr@ee.usc.edu).}%
\thanks{R.~Pedarsani is with the Department of Electrical and Computer
Engineering, University of California at Santa Barbara, Santa Barbara,
CA 93106, USA (e-mail: ramtin@ece.ucsb.edu).}%\thanks{This paper has supplementary downloadable material available at
%http://ieeexplore.ieee.org, provided by the authors. This includes a PDF
%containing Appendix A-E.}
}

% The paper headers
%\markboth{IEEE/ACM TRANSACTIONS ON NETWORKING}
%{Yang \MakeLowercase{\textit{et al.}}: Edge Computing in the Dark: Leveraging Contextual-Combinatorial Bandit and Coded Computing}

\maketitle
\begin{abstract}
With recent advancements in edge computing capabilities, there has been a significant increase in utilizing the edge cloud for event-driven and time-sensitive computations. However, large-scale edge computing networks can suffer substantially from unpredictable and unreliable computing resources which can result in high variability of service quality. 
%Thus, it is crucial to design efficient task scheduling policies that guarantee quality of service and the timeliness of computation queries. 
%
We consider the problem of computation offloading over unknown edge cloud networks with a sequence of timely computation jobs. 
Motivated by the MapReduce computation paradigm, we assume that each computation job can be partitioned to smaller Map functions which are processed at the edge, and the Reduce function is computed at the user after the Map results are collected from the edge nodes. We model the service quality %(success probability of returning result back to the user within deadline) 
of each edge device as function of context. %(collection of factors that affect edge devices). 
The user decides the computations to offload to each device with the goal of receiving a recoverable set of computation results in the given deadline. 
%Our goal is to design an efficient edge computing policy in the dark without the knowledge of the context or computation capabilities of each device. 
%
By leveraging the \emph{coded computing} framework in order to tackle failures or stragglers in computation, we formulate this problem using contextual-combinatorial multi-armed bandits (CC-MAB), and aim to maximize the cumulative expected reward. We propose an online learning policy called \emph{online coded edge computing policy}, which provably achieves asymptotically-optimal performance in terms of regret loss compared with the optimal offline policy for the proposed CC-MAB problem. 
In terms of the cumulative reward, it is shown that the online coded edge computing policy significantly outperforms other benchmarks via numerical studies. 
\end{abstract}
\section{Introduction}\label{sec:intro}
Recent advancements in edge cloud has enabled users to offload their computations of interest to the edge for processing. Specifically, there has been a significant increase in utilizing the edge cloud for event-driven and time-sensitive computations (e.g., IoT applications and cognitive services), in which the users increasingly demand timely services with deadline constraints, i.e., computations of requests have to be finished within specified deadlines. However, large-scale distributed computing networks can substantially suffer from unpredictable and unreliable computing infrastructure which can result in high variability of computing resources, i.e., service quality of the computing resources may vary over time. The speed variation has several causes including hardware failure, co-location of computation tasks, communication bottlenecks, etc \cite{zaharia2008improving,ananthanarayanan2013effective}. While edge computing has offered a novel framework for computing service provisioning, a careful design of task scheduling policy is still needed to guarantee the timeliness of task processing due to the increasing demand on real-time response of various applications and the unknown environment of the network.

To take advantage of the parallel computing resources for reducing the total latency, the applications are often modeled as a MapReduce computation model, i.e., the computation job can be partitioned to some smaller Map functions which can be distributedly processed by the edge devices. Since the data transmissions between the edge devices can result in large latency delay, it is often the case that the user computes the Reduce function on the results of the Map functions upon receiving the computation results of edge devices to complete the computation job.

In this paper, we study the problem of computation offloading over edge cloud networks with particular focus on \emph{unknown environment of computing resources} and \emph{timely computation jobs}. We consider a dynamic computation model, where a sequence of computation jobs needs to be computed over the (encoded) data that is distributedly stored at the edge nodes. More precisely, in an online manner, computation jobs with given deadlines are submitted to the edge network, i.e., each computation has to be finished within the given deadline. %Motivated by the MapReduce computation paradigm, each computation job can be partitioned to some smaller Map functions to better take advantage of parallelism of edge cloud. 
We assume the service quality (success probability of returning results back to the user in deadline) of each edge device is parameterized by a context (collection of factors that affect each edge device). The user aims at selecting edge devices from the available edge devices such that the user can receive a recoverable set of computation results in the given deadline. Our goal is then to design an efficient edge computing policy that maximizes the cumulative expected reward, where the expected reward collected at each round is a linear combination of the success probability of the computation and the amount of computational resources used (with negative sign).

One significant challenge in this problem is the joint design of (1) data storage scheme to provide robustness against unknown behaviors of edge devices; (2) computation offloading to 
edge device; and (3) an online learning policy for making the offloading decisions based on the past observed events. In our model, the computation capacities of the devices (e.g., how likely the computation can be returned to the user within the deadline) are unknown to the user. 

As the main contributions of the paper, we introduce a \emph{coded computing} framework in which the data is encoded and stored at the edge devices in order to provide robustness against unknown computation capabilities of the devices. The key idea of coded computing is to encode the data and design each worker's computation task such that the fastest responses of any $k$ workers out of total of $n$ workers suffice to complete the distributed computation, similar to classical coding theory where receiving any $k$ symbols out of $n$ transmitted symbols enables the receiver to decode the sent message. Under coded computing framework, we formulate a contextual-combinatorial multi-armed bandit (CC-MAB) problem for the edge computing problem, in which the Lagrange coding scheme is utilized for data encoding \cite{yu2019lagrange}. 

Then, we propose a policy called \emph{online coded edge computing policy}, and show that it achieves asymptotically optimal performance in terms of regret loss compared with the optimal offline policy for the proposed CC-MAB problem by the careful design of the policy parameters. To prove the asymptotic optimality of online coded edge computing policy, we divide the expected regret to three regret terms due to (1) exploration phases, (2) bad selections of edge devices in exploitation phases, and (3) good selections of edge devices in exploitation phases; then we bound these three regrets separately. 
%We show that the expected regret due to exploration phases is bounded by the order of policy parameters' product which can be chosen to be sublinear. By the Chernoff-Hoeffding inequality, we show that the expected regret due to bad selections of edge devices is bounded by a constant. By the property that edge devices have similar service qualities if they have similar contexts, we show that the expected regret due to good selections of edge devices has a sublinear bound.  
%
 
In addition to proving the asymptotic optimality of online coded edge computing policy, we carry out numerical studies using the real world scenarios of Amazon EC2 clusters. In terms of the cumulative reward, the results show that the online coded edge computing policy significantly outperforms other benchmarks.

In the following, we summarize the key contributions in this paper:
\begin{itemize}[leftmargin=*]
    \item We formulate the problem of coded edge computing using the CC-MAB framework.
    \item We propose online coded edge computing policy, which is provably asymptotically optimal. 
  \item We show that the online coded edge computing policy outperforms other benchmarks via numerical studies. 
\end{itemize}
\subsection{Related Prior Work}
Next, we provide a brief literature review that covers three main lines of work: task scheduling over cloud networks, coded computing, and the multi-armed bandit problem. 

In the dynamic task scheduling problem, jobs arrive to the network according to a stochastic process, and get scheduled dynamically over time. The first goal in task scheduling is to find a throughput-optimal scheduling policy (see e.g., \cite{eryilmaz2005stable}), i.e. a policy that stabilizes the network, whenever it can be stabilized. For example, Max-Weight scheduling, first proposed in \cite{tassiulas1992stability,dai2005maximum}, is known to be throughput-optimal for wireless networks, flexible queueing networks \cite{neely2005dynamic}, data centers networks \cite{maguluri2012stochastic} and dispersed computing networks \cite{yang2019communication}. Moreover, there have been many works which focus on task scheduling problem with deadline constraints over cloud networks (see e.g., \cite{hoseinnejhad2017deadline}). 

Coded computing broadly refers to a family of techniques that utilize coding to inject computation redundancy in order to alleviate the various issues that arise in large-scale distributed computing. In the past few years, coded computing has had a tremendous success in various problems, such as straggler mitigation and bandwidth reduction (e.g., \cite{lee2018speeding,li2018fundamental,dutta2016short,lee2017high,yu2017polynomial,tandon2017gradient,li2017coding,li2020coded,prakash2020codedgraph,yu2020straggler}). Coded computing has also been expanded in various directions, such as heterogeneous networks (e.g.,~\cite{reisizadeh2019coded}), partial stragglers (e.g.,~\cite{ferdinand2018hierarchical}), secure and private computing (e.g.,~\cite{chen2018draco,yu2019lagrange,yang2021codedboolean,so2019codedprivateml,so2020scalable}), distributed optimization (e.g.,~\cite{karakus2017straggler}), federated learning (e.g.,~\cite{so2020turbo,prakash2020coded,prakash2020hierarchical}), blockchains (e.g.,~\cite{yu2020codedtree,li2020polyshardTIFS}). In a dynamic setting,~\cite{yang2019timely,yang2019timelyisit} consider the coded computing framework with deadline constraints and develops a learning strategy that can adaptively assign computation loads to cloud devices. In this paper, we go beyond the two states Markov model considered in ~\cite{yang2019timely,yang2019timelyisit}, and make a substantial progress by combining the ideas of coded computing with contextual-combinatorial MAB, which is a more general framework that does not make any strong assumption (e.g., Markov model) on underlying model for the speed of edge devices.

The multi-armed bandit (MAB) problem has been widely studied to address the critical tradeoff between exploration and exploitation in sequential decision making under uncertainty of environment~\cite{lai1985asymptotically}. The goal of MAB is to learn the single optimal arm among a set of candidate arms of a priori unknown rewards by sequentially selecting one arm each time and observing its realized reward~\cite{auer2002finite}. Contextual bandit problem extends the basic MAB by considering the context-dependent reward functions~\cite{li2010contextual,sen2017contextual,shariff2018differentially}. The combinatorial bandit problem is another extension of the MAB by allowing multiple-play (select a set of arms) each time~\cite{gai2012combinatorial,li2019combinatorial}. The contextual-combinatorial MAB problem considered in this paper has also received much attention recently~\cite{chen2019task,li2016contextual,muller2016context,qin2014contextual}.
However,~\cite{li2016contextual,qin2014contextual} assume that the reward of an action is a linear function of the contexts different from the reward function considered in our paper. ~\cite{muller2016context} assumes the arm set is fixed throughout the time but the arms (edge devices) may appear and disappear across the time in edge networks.

\cite{chen2019task} considers a CC-MAB problem for the vehicle cloud computing, in which the tasks are deadline-constrained. However, the task replication technique used in \cite{chen2019task} is to replicate the "whole job" to multiple edge devices  without taking advantage of parallelism of computational resources. Coded computing is a more general technique which allows the dataset to be first partitioned to smaller datasets and then encoded such that each device has smaller computation compared to \cite{chen2019task}. Moreover, the success probability term (for receiving any $k$ results out of $n$ results) of reward function considered in our paper is more general than the success probability term (for receiving any $1$ result out of $n$ results) of reward function considered in \cite{chen2019task}.

\section{System Model}\label{sec:sys}
\subsection{Computation Model}\label{subsec:comp_model}
We consider an edge computing problem, in which a user offloads its computation to an edge network in an online manner, and the computation is executed by the edge devices. In particular, there is a given deadline for each round of computation, i.e., computation has to be finished within the given deadline. 

As shown in Fig. \ref{fig:system}, the considered edge network is composed of a user node and a set of edge devices. There is a dataset $X_1,X_2,\dots,X_k$ where each $X_j$ is an element in a vector space $\mathbb{V}$ over a sufficiently large finite field $\mathbb{F}$. Each edge device prestores the data which can be possibly a function of $X_1,X_2,\dots,X_k$.

Let $\{1,2,\dots,T\}$ be the index of the user's computation jobs received by the edge network over $T$ time slots.  In each round $t$ (or time slot in a discrete-time system), the user has a computation job denoted by function $g_t$. Especially, we assume that function $g_t$ can be computed by $$ g_t(X_1,X_2,\dots, X_k) = h_t(f_t(X_1),f_t(X_2),\dots,f_t(X_k))$$
where function $g_t$ and $f_t$ (with degree \textrm{deg}$(f_t)$) are multivariate polynomial functions with vector coefficients. In such edge network and motivated by a MapReduce setting, the user is interested in computing Map functions $f_t(X_1),f_t(X_2),\dots,f_t(X_k)$ in each round $t$ and the user computes Reduce function $h_t$ on those results of Map functions to obtain $g_t(X_1,X_2,\dots,X_k)$. 
\begin{remark}
We note that the considered computation model naturally appears in many machine learning applications which use gradient-type algorithms. For example, in linear regression problems given $\vec{y}_j$ which is the vector of observed labels for data $X_j$, each worker $j$ computes $f_t(X_j)$ $=X_j^{\top}(X_j\vec{w}_t-\vec{y}_j)$ which is the gradient of the quadratic loss function $\frac{1}{2}\|X_j\vec{w}_t-\vec{y}_j\|^2$ with respect to the weight vector $\vec{w}_t$ in round $t$. To complete the update $\vec{w}_{t+1} = g_t(X_1,\dots,X_k) = \vec{w}_t -\beta_t\sum^k_{j=1}f_t(X_j)$, the user has to collect the computation results $f_t(X_1), f_t(X_2),\dots f_t(X_k)$.

 Moreover, the considered computation model also holds for various edge computing applications. For example, in a mobile navigation application, the goal of user is to compute the fastest route to its destination. Given a dataset containing the map information and the traffic conditions over a period of time, edge devices compute map functions which output all possible routes between the two end locations. After collecting the intermediate results from edge devices, the user computes the best route.
\end{remark}

\begin{figure}[t]
    \centering
    \includegraphics[width =0.8\columnwidth]{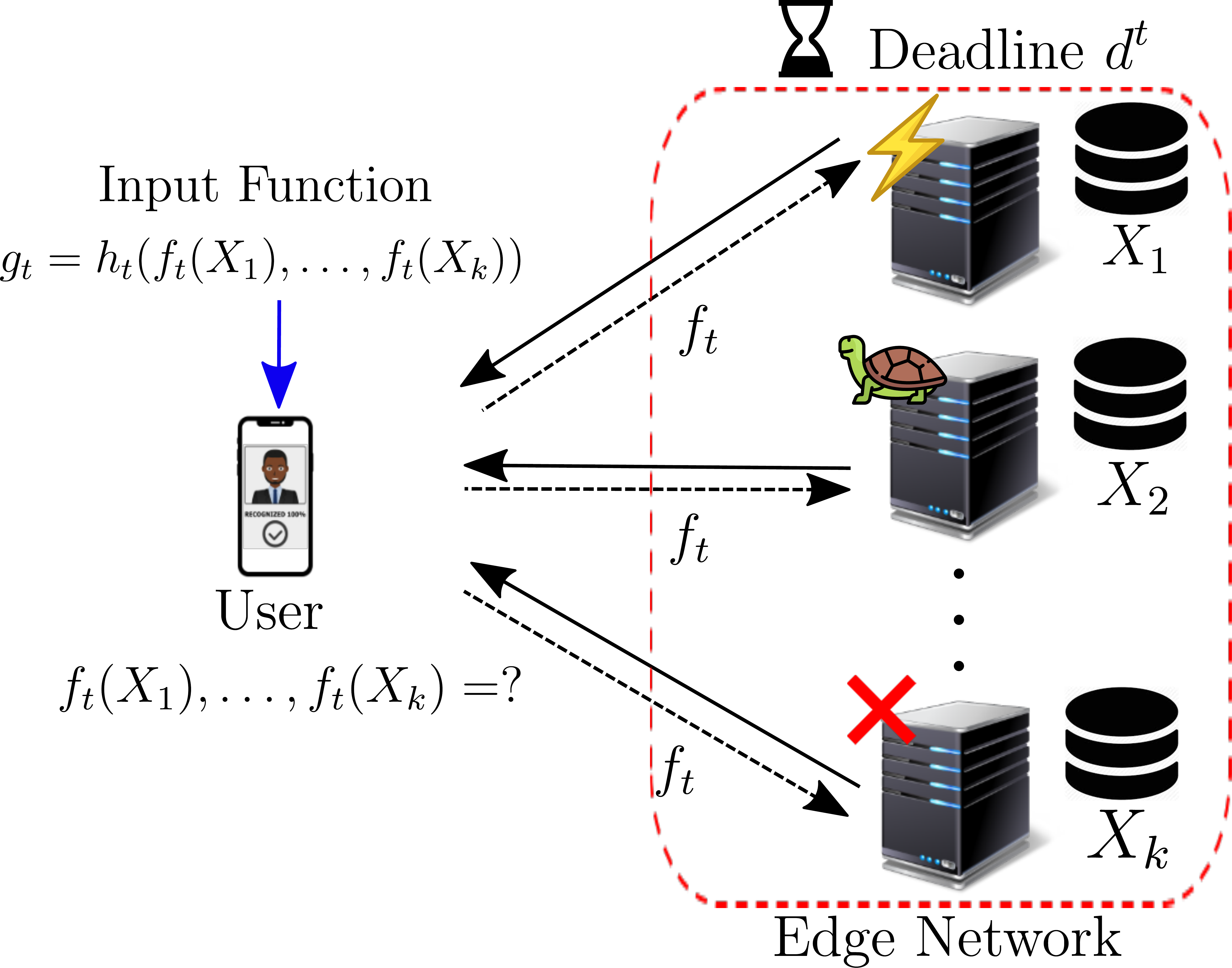}
    \caption{Overview of online computation offloading over an edge network with timely computation requests. In round $t$, the goal of user is to compute the Map functions $f_t(X_1),\dots,f_t(X_k)$ by the deadline $d^t$ using the edge devices.
    } 
    \label{fig:system}
\end{figure}
\subsection{Network Model}\label{subsec:net_model}
In an edge computing network, whether a computation result can be returned to the user depends on many factors. For example, the computation load of an edge device influences its runtime; the output size of the computation task affects the transmission delay, etc. Such factors are referred to as \emph{context} throughout the paper. The impact of each context on the edge devices is unknown to the user. More specifically, the computation service of each edge device is modeled as follows.

Let $\Phi_{T}$ be the context space of dimension $D_T$ includes $D_T$ different information of computation task, e.g., size of input/output, size of computation, and deadline, etc. Let $\Phi_S$ be the context space of dimension $D_S$ for edge devices which includes the information related to edge devices such as computation speed, bandwidth, etc. Let $\Phi = \Phi_T \times \Phi_S$ be the joint context space which is assumed to be bounded and thus can be defined by $\Phi = [0,1]^D$ and $D = D_T+D_S$ is the dimension of context space $\Phi$ without loss of generality. 

In each round $t$, let $\mathcal{V}^{t}$ denote the set of edge devices available to the user for computation, i.e., the available set of devices might change over time. Moreover, we denote by $b^t$ the budget (maximum number of devices to be used) in round $t$. The service delay (computation time plus transmission time) of each edge device $\nu$ is parameterized by a given context $\phi^t_{\nu} \in \Phi$. We denote by $c^t_{\nu}$ the service delay of edge device $\nu$, and $d^t$ the computation deadline in round $t$. Let $q^t_{\nu}= \mathbbm{1}_{\{c^t_{\nu} \leq d^t\}}$ be the indicator that the service delay of edge device $\nu$ is smaller than or equal to the given deadline $d^t$ in round $t$. Also, let $\mu(\phi^t_{\nu}) = \mathbb{E}[q^t_{\nu}] = \mathbb{P}(c^t_{\nu} \leq d^t)$ be the success probability that edge device $\nu$ returns the computation result back to the user within deadline $d^t$, and $\boldsymbol{\mu}^t = \{\mu(\phi^t_{\nu})\}_{\nu \in \mathcal{V}^t}$ be the collection of success probabilities of edge devices in round $t$. Let us illustrate the model through a simple example. 

\begin{example} \label{ex:net_model}
In~\cite{reisizadeh2019coded}, the shifted exponential distributions have been demonstrated to be a good fit for modeling the execution time of a node in cloud networks. Thus, we can model the success probability of an edge device as follows:
\begin{align*}
    \mu({\phi}^t) = \mathbb{P}(c^t \leq d^t) = 
    \begin{cases}
    1- e^{-\lambda^t(d^t-a^t)}&, \ d^t \geq a^t \\ 
    0&, \ a^t > d^t\geq 0,
    \end{cases}
\end{align*}
where the context space $\Phi$ consists of the deadline $d^t$, the shift parameter $a^t>0$, and the straggling parameter $\lambda^t >0$ associated with an edge device. 
\end{example}
\subsection{Problem Statement}
Let $\mathcal{V} = \{1,2,\dots,|\mathcal{V}|\}$ be the set of all edge devices in the network. Given context $\boldsymbol{\phi}^t = \{ \phi^t_{\nu} \}_{\nu \in \mathcal{V}^t}$ of the edge devices available to the user in round $t$, the goal of the user is to select a subset of edge devices from the available set of edge devices $\mathcal{V}^t \subseteq \mathcal{V}$, and decide what to be computed by each selected edge device, such that a \emph{recoverable} (or decodable as will be clarified later) set of computation results $f_t(X_1),\dots f_t(X_k)$ can be returned to the user within deadline $d^t$.

\section{Online Coded Edge Computing} \label{sec:LCC}
In this section, we introduce a coded computing framework for the edge computing problem, and formulate the problem as a contextual-combinatorial multi-armed bandit (CC-MAB) problem. Then, we propose a policy called \emph{online coded edge computing policy}, which is a context-aware learning algorithm. 
\subsection{Lagrange Coded Computing}
For the data storage of edge devices, we leverage a linear coding scheme called the Lagrange coding scheme \cite{yu2019lagrange} which is demonstrated to simultaneously provide resiliency, security, and privacy in distributed computing. We start with an illustrative example.

In each round $t$, we consider a computation job which consists of computing quadratic functions $f_t(X_j) = X_j^{\top}(X_j\vec{w}_t-\vec{y}_j)$ over available edge devices $\mathcal{V}^t = \{1,2,\dots,6\}$, where input dataset $X$ is partitioned to $X_1,X_2$. Then, we define function $m$ as follows:
\begin{align}
    m(z) \triangleq X_1\frac{z-1}{0-1}+X_2 \frac{z-0}{1-0} = z(X_2-X_1)+X_1,
\end{align}
in which $m(0) = X_1$ and $m(1)=X_2$. Then, we encode $X_1$ and $X_2$ to $\tilde{X}_{\nu} = m(\nu-1)$, i.e., $\tilde{X}_1=X_1$, $\tilde{X}_2=X_2$, $\tilde{X}_3=-X_1+2X_2$, $\tilde{X}_4=-2X_1+3X_2$, $\tilde{X}_5=-3X_1+4X_2$ and $\tilde{X}_6=-4X_1+5X_2$. Each edge device $\nu \in \{1,2,\dots,6\}$ prestores an encoded data chunk $\tilde{X}_{\nu}$ locally. If edge device $\nu$ is selected in round $t$, it computes $f_t(\tilde{X}_{\nu}) = \tilde{X}_{\nu}^{\top}(\tilde{X}_{\nu}\vec{w}_t-\vec{y}_{\nu})$ and returns the result back to the user upon its completion. We note that $f_t(\tilde{X}_{\nu}) = f_t(m(\nu-1))$ is an evaluation of the composition polynomial $f_t(m(z))$, whose degree at most $2$, which implies that $f_t(m(z))$ can be recovered by any $3$ results via polynomial interpolation. Then we have $f_t(X_1) = f_t(m(0))$ and $f_t(X_2) = f_t(m(1))$.

Formally, we describe Lagrange coding scheme as follows:\\
We first select $k$ distinct elements $\beta_1, \beta_2,\dots, \beta_k$ from $\mathbb{F}$, and let $m$ be the respective \emph{Lagrange interpolation polynomial} 
\begin{align}
    m(z) \triangleq \sum^{k}_{j=1}X_j \prod_{l \in [k]\backslash \{j\}}\frac{z-\beta_l}{\beta_j - \beta_l},
\end{align}
where $u: \mathbb{F} \rightarrow \mathbb{V}$ is a polynomial of degree $k-1$ such that $m(\beta_j) = X_j$. Recall that $\mathcal{V} = \cup^T_{t=1}\mathcal{V}^t$ which is the set of all edge devices. To encode input $X_1,X_2,\dots,X_k$, we select $|\mathcal{V}|$ distinct elements $\alpha_1,\alpha_2,\dots,\alpha_{|\mathcal{V}|}$ from $\mathbb{F}$, and encode $X_1,X_2,\dots,X_k$ to $\tilde{X}_v=m(\alpha_v)$ for all $v \in [|\mathcal{V}|]$, i.e.,
\begin{align}
    \tilde{X}_v = m(\alpha_v) \triangleq  \sum^{k}_{j=1}X_j \prod_{l \in [k]\backslash \{j\}}\frac{\alpha_v-\beta_l}{\beta_j - \beta_l}.
\end{align}
Each edge device $\nu \in \mathcal{V}$ stores $\tilde{X}_{\nu}$ locally. If edge device $\nu$ is selected in round $t$, it computes $f_t(\tilde{X}_{\nu})$ and returns the result back to the user upon its completion. Then , the optimal recovery threshold $Y^t$ using Lagrange coding scheme is 
\begin{align}
    Y^t = (k-1)\textrm{deg}(f_t)+1 \label{eq:recovery}
\end{align}
which guarantees that the computation tasks $f_t(X_1),\dots,f_t(X_k)$ can be recovered when the user receives any $Y^t$ results from the edge devices. The encoding of Lagrange coding scheme is oblivious to the	computation task $f_t$. Also, decoding and encoding process in Lagrange coding scheme rely on polynomial interpolation and evaluation which can be done efficiently.
\begin{remark}
We note that the data newly generated in edge device can be encoded and distributed to other devices at off-peak time. Especially, a key property of LCC is that the encoding process can be done incrementally, i.e., when there are some new added datasets, the update of encoded data can be done incrementally by encoding only on the new data instead of redoing the encoding on all the datasets. For example, let us consider the case that each data $X_j$ is represented by a vector. When there is a new generated data element $x_j$ added to each data $X_j$, we just encode new data elements $x_1,x_2,\dots,x_k$ to $\tilde{x}_1,\tilde{x}_2\dots$ and the new encoded data can be obtained by appending the new encoded data to old encoded data vectors $\tilde{X}_1,\tilde{X}_2,\dots$.
\end{remark}
\subsection{CC-MAB for Coded Edge Computing}
Now we consider a coded computing framework in which the Lagrange coding scheme is used for data encoding, i.e., each edge device $\nu$ prestores encoded data $\tilde{X}_{\nu}$. The encoding process is only performed once for dataset $X_1,\dots,X_k$. After Lagrange data encoding, the size of input data and computation of each user do not change, i.e., context $\phi^t_{\nu}$ of each edge device $\nu$ remains the same.  

More specifically, we denote by $\mathcal{A}^t$ the set of devices which are selected in round $t$ for computation. In each round $t$, the user picks a subset of devices $\mathcal{A}^t$ from all available devices $\mathcal{V}^t$, and we call $\mathcal{A}^t \subseteq \mathcal{V}^t$ the ``offloading decision''. The reward function $r(\mathcal{A}^t)$ achieved by offloading decision $\mathcal{A}^t$ is composed of the reward term and the cost term, which is defined as follows:
\begin{align}
    r(\mathcal{A}^t) = 
    \begin{cases}
    1-\eta |\mathcal{A}^t|, \ \text{if} \ \sum_{\nu \in \mathcal{A}^t} q^t_{\nu} \geq Y^t  \\
    -\eta |\mathcal{A}^t| , \ \text{if} \ \sum_{\nu \in \mathcal{A}^t} q^t_{\nu} < Y^t 
    \end{cases}
        \end{align}
where the term $|\mathcal{A}^t|$ captures the cost of using offloading decision $\mathcal{A}^t$ with the unit cost $\eta$ for using one edge device, and $Y^t$ is the optimal recovery threshold defined in (\ref{eq:recovery}). More precisely, the reward term is equal to $1$ if the total number of received results is greater than the optimal recovery threshold, i.e., $\sum_{\nu \in \mathcal{A}^t}q^t_{\nu} \geq Y^t$; otherwise the reward term is equal to $0$. On the other hand, the cost term is defined as $-\eta |\mathcal{A}^t|$ which is the cost of using $\mathcal{A}^t$. 

Then, the expected reward denoted by $u(\boldsymbol{\mu}^t,\mathcal{A}^t)$ in round $t$ can be rewritten as follows:
\begin{align}
    u(\boldsymbol{\mu}^t,\mathcal{A}^t) = &\sum^{|\mathcal{A}^t|}_{s=Y^t}\sum_{\mathcal{A} \subseteq \mathcal{A}^t,|\mathcal{A}|=s} \prod_{\nu \in \mathcal{A}}\mu(\phi^t_{\nu}) \prod_{\nu \in \mathcal{A}^t \backslash \mathcal{A}}(1- \mu(\phi^t_{\nu})) \nonumber\\  & -\eta |\mathcal{A}^t| \label{eq:reward}
\end{align}
where the first term of the expected reward of an offloading decision is the success probability that there are at least $Y^t$ computation results received by the user for LCC decoding. 

Consider an arbitrary sequence of computation jobs indexed by $\{1,2,\dots,T\}$ for which the user makes offloading decisions $\{\mathcal{A}^t\}^T_{t=1}$. To maximize the expected cumulative reward, we introduce a contextual-combinatorial multi-armed bandit (CC-MAB) problem for coded edge computing defined as follows:

\textbf{CC-MAB for Coded Edge Computing:}
\begin{align}
    & \max_{\{\mathcal{A}^t\}^T_{t=1}} \sum^T_{t=1} u(\boldsymbol{\mu}^t,\mathcal{A}^t) \\
    & \text{s.t.} \ \mathcal{A}^t \subseteq \mathcal{V}^t, \ |\mathcal{A}^t| \leq b^t, \ \forall t \in [T] \label{eq:budget}
\end{align}
where the constraint (\ref{eq:budget}) indicates that the number of edge devices in $\mathcal{A}^t$ cannot exceed the budget $b^t$ in round $t$. The proposed CC-MAB problem is equivalent to solving an independent subproblem in each round $t$ as follows:
\begin{align}
    \max_{\mathcal{A}^t}&\sum^{|\mathcal{A}^t|}_{s=Y^t}\sum_{\mathcal{A} \subseteq \mathcal{A}^t,|\mathcal{A}|=s} \prod_{\nu \in \mathcal{A}}\mu(\phi^t_{\nu}) \prod_{\nu \in \mathcal{A}^t \backslash \mathcal{A}}(1- \mu(\phi^t_{\nu}))  -\eta |\mathcal{A}^t| \nonumber\\
     \text{s.t.} \ & \mathcal{A}^t \subseteq \mathcal{V}^t; \ |\mathcal{A}^t| \leq b^t. \nonumber
\end{align}
\begin{remark}
We note that the proposed CC-MAB not only works for LCC but also for any other coding schemes. In this paper, we focus on LCC since LCC is a universal and optimal encoding technique for arbitrary multivariate polynomial computations.
\end{remark}
\subsection{Optimal Offline Policy}
We now assume that the success probability of each edge device $\nu \in \mathcal{V}^t$ is known to the user. In round $t$, to find the optimal $\mathcal{A}^{t*}$, we present the following intuitive lemma proved in Appendix \ref{proof_lemma_optimalset}. 
\begin{lemma} \label{lemma:optimalset}
Without loss of generality, we assume $\mu(\phi^t_{1}) \geq \mu(\phi^t_{2})  \geq \dots \geq \mu(\phi^t_{|\mathcal{V}^t|})$ in round $t$. Considering all possible sets $\mathcal{A}^t_g \subseteq \mathcal{V}^t$ with fixed cardinality $n_g$, the optimal $\mathcal{A}^{t*}_g$ with cardinality $n_g$ that achieves the largest expected reward $u(\boldsymbol{\mu}^t,\mathcal{A}^t_g)$ is 
\begin{align}
    \mathcal{A}^{t*}_g =\{1,2,\dots,n_g\}
\end{align}
which represents the set of $n_g$ edge devices having largest success probability $\mu(\phi^t_{\nu})$ among all the edge devices. 
\end{lemma}
By Lemma \ref{lemma:optimalset}, to find the optimal set $\mathcal{A}^{t*}$, we can only focus on finding the optimal size of $\mathcal{A}^t$. Since there are only $b^t$ choices for size of $|\mathcal{A}^t|$ (i.e., $1,2,\dots,{b^t}$), this procedure can be done by a linear search with the complexity linear in the number of edge devices $|\mathcal{V}^t|$. We present the optimal offline policy in Algorithm \ref{alg:optimal_oracle}.

\begin{remark}
We note that the expected reward function considered in \cite{chen2019task} is a submodular function, which can be maximized by a greedy algorithm. However, the expected reward function defined in \eqref{eq:reward} is more general which cannot be maximized by the greedy algorithm. More specifically, one can show that the expected reward defined in equation \eqref{eq:reward} is not submodular by checking the property of submodular functions, i.e., for all possible subsets $\mathcal{A} \subseteq \mathcal{B} \subseteq \mathcal{V}$, $u(\boldsymbol{\mu},\{\nu\} \cup \mathcal{A}) - u(\boldsymbol{\mu},\mathcal{A})\geq u(\boldsymbol{\mu},\{\nu\} \cup \mathcal{B}) - u(\boldsymbol{\mu},\mathcal{B})$ does not hold. Without the property of submodularity, Lemma \ref{lemma:optimalset} enables us to maximize equation \eqref{eq:reward} by a linear search.
\end{remark}
Let $\{\mathcal{A}^t\}^T_{t=1}$ be the offloading decisions derived by a certain policy. The performance of this policy is evaluated by comparing its loss with respect to the optimal offline policy. This loss is called the regret of the policy which is formally defined as follows:
\begin{align}
    R(T) & = \mathbb{E}\big[ \sum^T_{t=1}r(\mathcal{A}^{t*})-r(\mathcal{A}^t)\big]\\
    & = \sum^T_{t=1} u(\boldsymbol{\mu}^t,\mathcal{A}^{t*}) - u(\boldsymbol{\mu}^t,\mathcal{A}^{t}).
\end{align}
In general, the user does not know in advance the success probabilities of edge devices due to the uncertainty of the environment of edge network. In the following subsection, we will propose an online learning policy for the proposed CC-MAB problem which enables the user to learn the success probabilities of edge devices over time by observing the service quality of each selected edge device, and then make offloading decisions adaptively.
\begin{algorithm}[t]
\SetAlgoLined
\textbf{Input:} $\mathcal{V}^t,b^t,Y^t,\mu(\phi^t_{\nu}), \nu \in \mathcal{V}^t$\;
\textbf{Initialization:} $\mathcal{A} = \emptyset$, $\mathcal{A}_{\textrm{opt}} = \emptyset$, $u_{\textrm{opt}} = 0$\;
 Sort $\boldsymbol{\mu}^t:$ $\mu(\phi^t_{1}) \geq \mu(\phi^t_{2})  \geq \dots \geq \mu(\phi^t_{|\mathcal{V}^t|})$\;
 $\mathcal{A} \leftarrow \{1,2,\dots,Y^t\}$ \;
 $\mathcal{A}_{\textrm{opt}} \leftarrow \{1,2,\dots,Y^t\}$ \;
 $u_{\textrm{opt}} \leftarrow u(\boldsymbol{\mu}^t,\mathcal{A})$\;
 \For{$z \leftarrow Y^t+1$ \KwTo $b^t$}{
 $\mathcal{A} \leftarrow \mathcal{A} \cup \{z\}$\;
  \If{$u(\boldsymbol{\mu}^t,\mathcal{A})>u_{\textrm{opt}}$}{
  $\mathcal{A}_{\textrm{opt}} \leftarrow \mathcal{A}$\;
  $u_{\textrm{opt}} \leftarrow u(\boldsymbol{\mu}^t,\mathcal{A})$
  }
                                                                        }
 \Return $\mathcal{A}_{\textrm{opt}}$
 \caption{Optimal Offline Policy}\label{alg:optimal_oracle}
\end{algorithm}
\subsection{Online Coded Edge Computing Policy}
Now, we describe the proposed online edge computing policy. The proposed policy has two parameters $h_T$ and $K(t)$ to be designed, where $h_T$ decides how we partition the context space, and $K(t)$ is a deterministic and monotonically increasing function, used to identify the under-explored context. The proposed online coded edge computing policy (see Algorithm \ref{alg:online_LCC}) is performed as follows:

\textbf{Initialization Phase:} Given parameter $h_T$, the proposed policy first creates a partition denoted by $\mathcal{P}_T$ for the context space $\Phi$, which splits $\Phi$ into $(h_T)^D$ sets. Each set is a $D$-dimensional hypercube of size $\frac{1}{h_T} \times \dots \times \frac{1}{h_T}$. For each hypercube $p \in \mathcal{P}_T$, the user keeps a counter $C^t(p)$ which is the number of selected edge devices that have context $\phi^t_{\nu}$ in hypercube $p$ before round $t$. Moreover, the policy also keeps an estimated success probability denoted by $\hat{\mu}^t(p)$ for each hypercube $p$. Let $\mathcal{Q}^t(p) = \{q^{\tau}_{\nu}:\phi^{\tau}_{\nu}\in p,\nu \in \mathcal{A}^{\tau},\tau = 1,\dots,t-1 \}$ be the set of observed indicators (successful or not) of edge devices with context in $p$ before round $t$. Then, the estimated success probability for edge devices with context $\phi^t_{\nu} \in p$ is computed by $\hat{\mu}^t(p) = \frac{1}{C^t(p)}\sum_{q \in \mathcal{Q}^t(p)}q$.

In each round $t$, the proposed policy has the following phases: 

\textbf{Hypercube Identification Phase:} Given the contexts of all available edge devices $\boldsymbol{\phi}^t = \{\phi^t_{\nu}\}_{\nu \in \mathcal{V}^t}$, the policy determines the hypercube $p^t_{\nu} \in \mathcal{P}_T$ for each context $\phi^t_{\nu}$ such that $\phi^t_{\nu}$ is in $p^t_{\nu}$. We denote by $\boldsymbol{p}^t = \{p^t_{\nu}\}_{\nu \in \mathcal{V}^t}$ the collection of these identified hypercubes in round $t$. To check whether there exist hypercubes $p \in \boldsymbol{p}^t$ that have not been explored sufficiently, we define the under-explored hypercubes in round $t$ as follows:
\begin{align}
    \mathcal{P}^{ue,t}_T = \{p \in \mathcal{P}_T: \exists \nu \in \mathcal{V}^t,\phi^t_{\nu} \in p,C^t(p) \leq K(t)\}.
\end{align}
Also, we denote by $\mathcal{V}^{ue,t}$ the set of edge devices which fall in the under-explored hypercubes, i.e., $\mathcal{V}^{ue,t} = \{\nu \in \mathcal{V}^t: p^t_{\nu} \in \mathcal{P}^{ue,t}_T\}$. Depending on  $\mathcal{V}^{ue,t}$ in round $t$, the proposed policy then either enters an exploration phase or an exploitation phase. 

 \textbf{Exploration Phase:} If $\mathcal{V}^{ue,t}$ is non-empty, the policy enters an exploration phase. If set $\mathcal{V}^{ue,t}$ contains at least $b^t$ edge devices (i.e., $|\mathcal{V}^{ue,t}| \geq b^t$), then the policy randomly selects $b^t$ edge devices from $\mathcal{V}^{ue,t}$. If $\mathcal{V}^{ue,t}$ contains fewer than $b^t$ edge devices ($|\mathcal{V}^{ue,t}| < b^t$), then the policy selects all edge devices from $\mathcal{V}^{ue,t}$. To fully utilize the budget $b^t$, the remaining $(b^t- |\mathcal{V}^{ue,t}|)$ ones are picked from the edge devices with the highest estimated success probability among the remaining edge devices in $\mathcal{V}^t \backslash \mathcal{V}^{ue,t}$.

\textbf{Exploitation Phase:}  If $\mathcal{V}^{ue,t}$ is empty, the policy enters an exploitation
phase and it selects $\mathcal{A}^t$ using the optimal offline policy based on the estimated success probabilities $\hat{\boldsymbol{\mu}}^t = \{\hat{\mu}^t(p^t_{\nu})\}_{\nu \in \mathcal{V}^t}$.

\textbf{Update Phase:} After selecting the edge devices, the proposed policy observes whether each selected edge device returns the result within the deadline; then, it updates $\hat{\mu}^t(p^t_{\nu})$ and $C^t(p^t_{\nu})$ of each hypercube $p^t_{\nu} \in \mathcal{P}_T$.

The following example illustrates how the policy works given parameters $h_T$ and $K(t)$. 
\begin{example} \label{ex:policy}
Consider the edge computing network in which the success probability of an edge device is defined by a shifted exponential distribution as defined in Example \ref{ex:net_model}. It can be shown that the H\"{o}lder condition with $\alpha = 1$ holds. Then, we have parameters $h_T = \lceil T^{\frac{1}{6}}\rceil $ and $K(t) = t^{\frac{1}{3}}\log{(t)}$. We assume that the online coded edge computing policy is run over time horizon $T = 1000$. Then, we have $h_T = 4$. Before running the policy, we create $\mathcal{P}_T$ by partitioning the domain of each context (i.e., deadlines, shift parameters and straggling parameters) into $h_T = 4$ intervals, which generates totally $64$ sets. We keep a counter $C^t(p)$ for each generated hypercube $p \in \mathcal{P}_T$. In hypercube identification phase, if there exists edge device $\nu$ with $\phi^t_{\nu}$ such that $\phi^t_{\nu}$ is located in hypercube $p$ and the counter $C^t(p)$ is smaller than $K(t)$, then $p$ is the under-explored hypercube. The policy will proceed to either exploration phase or exploitation phase depending on whether under-explored hypercube exists.   
\end{example}
\begin{algorithm}[t] 
\SetAlgoLined
\textbf{Input:} $T,h_T,K(t)$\;
 \textbf{Initialization:} $\mathcal{P}_T$; $C(p) = 0$,  $\hat{\mu}(p)=0$, $\forall p \in \mathcal{P}_T$ \;
 \For{$t\leftarrow 1$ \KwTo $T$}{
 Observe edge device $\mathcal{V}^t$ and contexts $\boldsymbol{\phi}^t$ \;
 Find $\boldsymbol{p}^t = \{p^t_{\nu}\}_{\nu \in \mathcal{V}^t}$, $p^t_{\nu}\in \mathcal{P}_T$ such that $\phi^t_{\nu} \in p^t_{\nu}$ \;
 Identify $\mathcal{P}^{ue,t}$ and $\mathcal{V}^{ue,t}$\;
 \eIf{$\mathcal{P}^{ue,t} \neq \emptyset$}{
 \eIf{$|\mathcal{V}^{ue,t}|\geq b^t$}{$\mathcal{A}^t \leftarrow$ randomly pick $b^t$ edge devices in $\mathcal{V}^{ue,t}$\;}
 {$\mathcal{A}^t \leftarrow $ pick all edge devices in $\mathcal{V}^{ue,t}$ and other $(b^t-|\mathcal{V}^{ue,t}|)$ ones with the largest $\hat{\mu}(p^t_{\nu})$ in $\mathcal{V}^t \backslash \mathcal{V}^{ue,t}$}}
 {$\mathcal{A}^t \leftarrow $ obtained by Algorithm \ref{alg:optimal_oracle} based on $\hat{\boldsymbol{\mu}}^t$ and $b^t$}
 \For{each edge device $\nu \in \mathcal{A}^t$}{
 Observe $q^t_{\nu}$ of edge device $\nu$\;
 
 Update $\hat{\mu}(p^{t}_{\nu}) = \frac{ \hat{\mu}(p^{t}_{\nu})C(p^{t}_{\nu})+q^t_{\nu}}{C(p^t_{\nu})+1}$\;
 Update $C(p^{t}_{\nu}) = C(p^{t}_{\nu})+1$\;
 }
                                                                        }
 \caption{Online Coded Edge Computing Policy}\label{alg:online_LCC}
\end{algorithm}
\section{Asymptotic Optimality of Online Coded Edge Computing Policy}\label{sec:optimality}
In this section, by providing the design of policy parameters $h_T$ and $K(t)$, we show that the online coded edge computing policy achieves a sublinear regret in the time horizon $T$ which guarantees an asymptotically optimal performance, i.e., $\lim_{T \rightarrow \infty}\frac{R(T)}{T} = 0$.

To conduct the regret analysis for the proposed CC-MAB problem, we make the following assumption on the success probabilities of edge devices in which the devices' success probabilities are equal if they have the same contexts. This natural property is formalized by the H\"{o}lder condition defined as follows:
\begin{assumption}[H\"{o}lder Condition] \label{ass:holder}
A real function $f$ on $D$-dimensional Euclidean space satisfies a H\"{o}lder condition, when there exist $L>0$ and $\alpha > 0$ for any two contexts $\phi, \phi' \in \Phi$, such that $|f(\phi) - f(\phi^{'})| \leq L \parallel \phi - \phi^{'}\parallel^{\alpha}$, where $\parallel \cdot \parallel$ is the Euclidean norm.
\end{assumption}
Under Assumption \ref{ass:holder}, we choose parameters $h_T = \lceil T^{\frac{1}{3\alpha+D}}\rceil$ for the partition of context space $\Phi$ and $K(t) = t ^{\frac{2\alpha}{3\alpha +D}}\log{(t)}$ in round $t$ for identifying the under-explored hypercubes of the context. We present the following theorem which shows that the proposed online coded edge computing policy has a sublinear regret upper bound.
\begin{theorem}[Regret Upper Bound] \label{thm:LCC}
Let $K(t) = t ^{\frac{2\alpha}{3\alpha +D}}\log{(t)}$ and $h_T = \lceil T^{\frac{1}{3\alpha+D}}\rceil$. If the H\"{o}lder condition holds, the regret $R(T)$ is upper-bounded as follows: 
\begin{align}
      R&(T)\leq (1+\eta B) 2^D (T^{\frac{2\alpha+D}{3\alpha+D}}\log{(T)}+T^{\frac{D}{3\alpha+D}}) \nonumber \\& + (1+\eta B) B \frac{\pi^2}{3} \sum^B_{k=1} \binom{|\mathcal{V}|}{k} \nonumber\\ &+ (3LD^{\frac{\alpha}{2}} + \frac{6\alpha+2D}{2\alpha+D})BM T^{\frac{2\alpha+D}{3\alpha+D}}, \nonumber
\end{align}
where $B = \max_{1\leq t \leq T} b^t$ and $M = \max_{1\leq t \leq T}\binom{B-1}{Y^t-1}$. The dominant order of the regret $R(T)$ is $O( T^{\frac{2\alpha+D}{3\alpha+D}}\log{(T)})$ which is sublinear to $T$.
\end{theorem}
\begin{proof}
%See Appendix \ref{proof_thm}.
We first define the following terms. For each hypercube $p \in \mathcal{P}_T$, we define $\overline{\mu} = \sup_{\phi \in p} \mu(\phi)$ and $\underline{\mu} = \inf_{\phi \in p} \mu(\phi)$ as the best and worst success probabilities over all contexts $\phi \in p$. Also, we define the context at center of a hypercube $p$ as $\tilde{\phi}_p$ and its success probability $\tilde{\mu}(p) = \mu(\tilde{\phi}_p)$. Given a set of available edge devices $\mathcal{V}^t$, the corresponding context set $\Phi^t = \{\phi^t_{\nu}\}_{\nu \in \mathcal{V}^t}$ and the corresponding hypercube set $\mathcal{P}^t = \{p^t_{\nu}\}_{\nu \in \mathcal{V}^t}$ for each round $t$, we also define $\overline{\boldsymbol{\mu}}^t = \{\overline{\mu}(p^t_{\nu})\}_{\nu \in \mathcal{V}^t}$, $\underline{\boldsymbol{\mu}}^t = \{\underline{\mu}(p^t_{\nu})\}_{\nu \in \mathcal{V}^t}$ and $\tilde{\boldsymbol{\mu}}^t = \{\tilde{\mu}(p^t_{\nu})\}_{\nu \in \mathcal{V}^t}$. For each round $t$, we define set $\tilde{\mathcal{A}}^t$ which satisfies 
\begin{align}
    \tilde{\mathcal{A}}^t = \argmax_{\mathcal{A} \subseteq \mathcal{V}^t, |\mathcal{A}|\leq b^t} u(\tilde{\boldsymbol{\mu}}^t,\mathcal{A})
\end{align}
We then use set $\tilde{\mathcal{A}}^t$ to identify the set of edge device which are bad to select. We define 
\begin{align}
    \mathcal{L}^t = \big\{G: G \subseteq \mathcal{V}^t, |G| \leq b^t, u(\underline{\boldsymbol{\mu}}^t,\tilde{\mathcal{A}}^t)- u(\overline{\boldsymbol{\mu}}^t,G) \geq At^{\theta} \big\} \nonumber
\end{align}
to be the set of \textit{suboptimal subsets of arms} for hypercube set $\mathcal{P}^t$, where $A> 0$ and $\theta<0$ are the parameters which will be used later in the regret analysis. We call a subset $G \in \mathcal{L}^t$ \textit{suboptimal}  and $\mathcal{A}^t_{b^{-}} \backslash \mathcal{L}^t$ \textit{near-optimal} for $\mathcal{P}^t$, where $\mathcal{A}^t_{b^{-}}$ denotes the subset of $\mathcal{V}^t$ with size less than $b^t$. Then the expected regret $R(T)$ can be divided into three summands:
\begin{align}
    R(T) = \mathbb{E}[R_e(T)] + \mathbb{E}[R_s(T)] + \mathbb{E}[R_n(T)],
\end{align}
where $\mathbb{E}[R_e(T)]$ is the regret due to exploration phases and $\mathbb{E}[R_s(T)]$ and $\mathbb{E}[R_n(T)]$ both correspond to regret in exploitation phases: $\mathbb{E}[R_s(T)]$ is the regret due to suboptimal choices, i.e., the subsets of edge devices from $\mathcal{L}_t$ are selected; $\mathbb{E}[R_n(T)]$ is the regret due to near-optimal choices, i.e., the subsets of edge devices from $\mathcal{A}^t_{b-} \backslash \mathcal{L}^t$. In the following, we prove that each of the three summands is bounded.

First, the following lemma (see the proof in Appendix \ref{proof_lemma_Re}) gives a bound for $\mathbb{E}[R_e(T)]$, which depends on the choice of two parameters $z$ and $\gamma$.
\begin{lemma}(Bound for $\mathbb{E}[R_e(T)]$). \label{lemma:R_e}Let $K(t) = t^z\log{(t)}$ and $h_T = \lceil T^{\gamma} \rceil$, where $0 < z < 1$ and $0 < \gamma < \frac{1}{D}$. If the algorithm is run with these parameters, the regret $E[R_e(T)]$ is bounded by
\begin{align}
    \mathbb{E}[R_e(T)] \leq (1+\eta B) 2^D (T^{z+\gamma D} \log{(T)}+T^{\gamma D})
\end{align}
where $B = \max_{1 \leq t \leq T}b^t$.
\end{lemma}
Next, the following lemma (see the proof in Appendix \ref{proof_lemma_Rs}) gives a bound for $\mathbb{E}[R_s(T)]$, which depends on the choice of  $z$ and $\gamma$ with an additional condition of these parameters which has to be satisfied.
\begin{lemma}(Bound for $\mathbb{E}[R_s(T)]$). \label{lemma:R_s}
Let $K(t) = t^z\log{(t)}$ and $h_T = \lceil T^{\gamma} \rceil$, where $0 < z < 1$ and $0 < \gamma < \frac{1}{D}$. If the algorithm is run with these parameters, Assumption \ref{ass:holder} holds, and the additional condition $2BMt^{-\frac{z}{2}}\leq At^{\theta}$ is satisfied for all $1 \leq t \leq T$, the regret $\mathbb{E}[R_s(T)]$ is bounded by
\begin{align}
    \mathbb{E}[R_s(T)] \leq (1+\eta B)B\frac{\pi^2}{3} \sum^B_{k=1}\binom{|\mathcal{V}|}{k},
\end{align}
where $B = \max_{1 \leq t \leq T}b^t$, and $M = \max_{1 \leq t \leq T}\binom{B-1}{Y^t-1}$. 
\end{lemma}
Lastly, the following lemma (see the proof in Appendix \ref{proof_lemma_Rn}) gives a bound for $\mathbb{E}[R_n(T)]$, which depends on the choice of $z$ and $\gamma$.
\begin{lemma}(Bound for $\mathbb{E}[R_n(T)]$). \label{lemma:R_n} Let $K(t) = t^z\log{(t)}$ and $h_T = \lceil T^{\gamma} \rceil$, where $0 < z < 1$ and $0 < \gamma < \frac{1}{D}$. If the algorithm is run with these parameters and Assumption \ref{ass:holder} holds, the regret $\mathbb{E}[R_n(T)]$ is bounded by
\begin{align}
    \mathbb{E}[R_n(T)] \leq 3BMLD^{\frac{\alpha}{2}}T^{1-\gamma \alpha} + \frac{A}{1+\theta}T^{1+\theta}.
\end{align}
where $B = \max_{1 \leq t \leq T} b^t$ and $M = \max_{1\leq t \leq T}\binom{B-1}{Y^t-1}$. 
\end{lemma}
Now, let $K(t) = t^z\log{(t)}$ and $h_T = \lceil T^{\gamma} \rceil$, where $0 < z < 1$ and $0 < \gamma < \frac{1}{D}$; let $H(t) = BMt^{-\frac{z}{2}}$. Also, we assume that Assumption \ref{ass:holder} holds and the additional condition $2BMt^{-\frac{z}{2}}\leq At^{\theta}$ is satisfied for all $1 \leq t \leq T$. By Lemma \ref{lemma:R_e}, \ref{lemma:R_s}, and \ref{lemma:R_n}, the regret $R(T)$ is bounded as follows:
\begin{align}
    R(T) \leq & (1+\eta B) 2^D (T^{z+\gamma D}\log{(T)}+T^{\gamma D}) \nonumber \\& + (1+\eta B) B \frac{\pi^2}{3} \sum^B_{k=1} \binom{|\mathcal{V}|}{k}\nonumber\\ &+ 3BMLD^{\frac{\alpha}{2}} T^{1- \alpha \gamma} + \frac{A}{1+\theta} T^{1+\theta}. 
\end{align}
Now, we select the parameters $z,\gamma,A,\theta$ according to the following values $z =\frac{2\alpha}{3\alpha+D} \in (0,1)$, $\gamma = \frac{1}{3\alpha+D} \in (0,\frac{1}{D})$, $\theta = -\frac{\alpha}{3\alpha+D}$ and $A = 2BM$. It is clear that condition $2BMt^{-\frac{z}{2}}\leq At^{\theta}$ is satisfied. Then, the regret $R(T)$ can be bounded as follows:
\begin{align}
        R&(T)  \leq  (1+\eta B) 2^D (T^{\frac{2\alpha+D}{3\alpha+D}}\log{(T)}+T^{\frac{D}{3\alpha+D}}) \nonumber \\& + (1+\eta B) B \frac{\pi^2}{3} \sum^B_{k=1} \binom{|\mathcal{V}|}{k} \nonumber\\ & + (3LD^{\frac{\alpha}{2}} + \frac{6\alpha+2D}{2\alpha+D}) BMT^{\frac{2\alpha+D}{3\alpha+D}},
\end{align}
which has the dominant order $O(T^{\frac{2\alpha+D}{3\alpha+D}}\log{(T)})$.
\end{proof}

\begin{remark}  
Based on Assumption~\ref{ass:holder}, the parameters $h_T$ and $K(t)$ are designed such that the regret achieved by the policy is sublinear as stated in Theorem~\ref{thm:LCC}. In the following, we provide some intuitions behind the choices of $h_T$ and $K(T)$. We first assume that the parameters are chosen as $h_T = \lceil T^{\gamma}\rceil$ and $K(t) = t^z \log{(t)}$, in which $\gamma$ and $z$ are designed later. In the proof of Lemma~\ref{lemma:R_e}, to bound $\mathbb{E}[R_e(T)]$, our main task is designing $T^{z+\gamma D}\log{(T)} + T^{\gamma D}$ to be a sublinear term. In the proof of Lemma~\ref{lemma:R_s}, one of key steps is to bound $\text{Pr}(V^t_{G},W^t)$ by the term of $t^{-2}$ (see equation (42)) such that $\mathbb{E}[R_s(T)]$ is bounded by the term of $\sum^{\infty}_{t=1} t^{-2}$ which converges to a constant (see equation (46)). In particular, we first bound $\text{Pr}(E_1)$ and $\text{Pr}(E_2)$ by the terms of $\exp{\frac{-2t^z\log{(t)}H(t)^2}{B^2M^2}}$ (see equation (37)), and choose $H(t)$ to be $BMt^{-\frac{z}{2}}$. By Lemma~\ref{lemma:R_n}, we bound $\mathbb{E}[R_n]$ by $3BMLD^{\frac{\alpha}{2}}T^{1-\alpha\gamma} + \frac{A}{1+\theta} T^{1+\theta}$ which can be sublinear by selecting the appropriate $\gamma$ and $z$. By carefully selecting parameters $h_T = \lceil T^{\frac{1}{3\alpha +D}}\rceil $ and $K(t) = t^{\frac{2\alpha}{3\alpha +D}}$, the regret upper bound is shown to be subliear in Theorem~\ref{thm:LCC}.
\end{remark}

\section{Experiments}
In this section, we demonstrate the impact of the online coded edge computing policy by simulation studies. In particular, we carry out extensive simulations using the shifted exponential models which have been demonstrated to be a good model for Amazon EC2 clusters~\cite{reisizadeh2019coded}.

Given a dataset partitioned to $X_1,X_2,\dots,X_{5}$, we consider the linear regression problem using the gradient algorithm. It computes the gradient of quadratic loss function $\frac{1}{2}\|X_j\vec{w}_t-\vec{y}_j\|^2$ with respect to the weight vector $\vec{w}_t$ in round $t$, i.e., $f_t(X_j)=X_j^{\top}(X_j\vec{w}_t-\vec{y}_j)$ for all $1\leq j\leq5$. The computation is executed over a set of edge devices $\mathcal{V}$, where each edge device $\nu \in \mathcal{V}$ stores an encoded data chunk $\tilde{X}_{\nu}$ using Lagrange coding scheme. In such setting, we have the optimal recovery threshold $Y^t = 9$. The penalty parameter $\eta$ is $0.01$.

Motivated by the distribution model proposed in~\cite{reisizadeh2019coded} for total execution time in cloud networks, we model the success probability of each edge device $\nu \in \mathcal{V}$ as a shifted exponential function defined as follows:
\begin{align}\label{eq:mu}
    \mu({\phi}^t_{\nu}) = \mathbb{P}(c^t_{\nu}\leq d^t) = 
    \begin{cases}
    1- e^{-\lambda^t_{\nu}(d^t-a^t_{\nu})}&, \ d^t \geq a^t_{\nu}, \\ 
    0&, \ a^t_{\nu} > d^t\geq 0,
    \end{cases}
\end{align}
where the context of each edge device consists of the deadline $d^t$, the shift parameter $a^t_{\nu}>0$, and the straggling parameter $\lambda^t_{\nu} >0$ associated with edge device $\nu$. Under this model, the dimension of context space $D$ is $3$. Moreover, for function $\mu$ defined in \eqref{eq:mu}, it can be shown that the H\"{o}lder condition with $\alpha = 1$ holds. Thus, we run the online coded edge computing policy with parameters $h_T = \lceil T^{\frac{1}{6}}\rceil $ and $K(t) = t^{\frac{1}{3}}\log{(t)}$.

%In each round $t$, the deadline $d^t\in [3,5]$ (sec), the shift parameter $a^t_{\nu} \in [1,2]$ (sec), and the straggling parameter $\lambda^t_{\nu} \in [0.2, 0.8]$ (1/sec) are chosen uniformly at random.  The budget $b^t$ is fixed to $12$ throughout the time for simplicity. The penalty parameter $\eta$ is $0.01$.

By the empirical analysis in~\cite{reisizadeh2019coded}, the instance of type \texttt{r4.2xlarge} is shown to have the shift parameter $a = 1.37$ and the straggling parameter $\lambda = 120$. And, the instance of type \texttt{r4.xlarge} has the shift parameter $a = 2$ and the straggling parameter $\lambda = 115$. 
 Based on the real-world parameters for Amazon EC2 clusters, the deadline $d^t\in [d_{\text{min}},d_{\text{max}}]$ (sec), the shift parameter $a^t_{\nu} \in [1.37,2]$ (sec), and the straggling parameter $\lambda^t_{\nu} \in [ 115, 120]$ (1/sec) are chosen uniformly at random in each round $t$. We consider the following four scenarios for the simulations:
\begin{itemize}[leftmargin=*]
    \item \textbf{Scenario 1:} $|\mathcal{V}| =20$, $(d_{\text{min}},d_{\text{max}}) = (1,2)$, and $b^t = 12$.
    \item  \textbf{Scenario 2:} $|\mathcal{V}| =15$, $(d_{\text{min}},d_{\text{max}}) = (1,2)$, and $b^t = 12$.
    \item  \textbf{Scenario 3:} $|\mathcal{V}| =20$,  
    $(d_{\text{min}},d_{\text{max}}) = (1,2)$, and $b^t = 15$.
    \item \textbf{Scenario 4:} $|\mathcal{V}| =20$, $(d_{\text{min}},d_{\text{max}}) = (0.5,3)$, and $b^t = 12$.
\end{itemize}
\begin{figure}[t]
  \centering
    \includegraphics[width = \columnwidth]{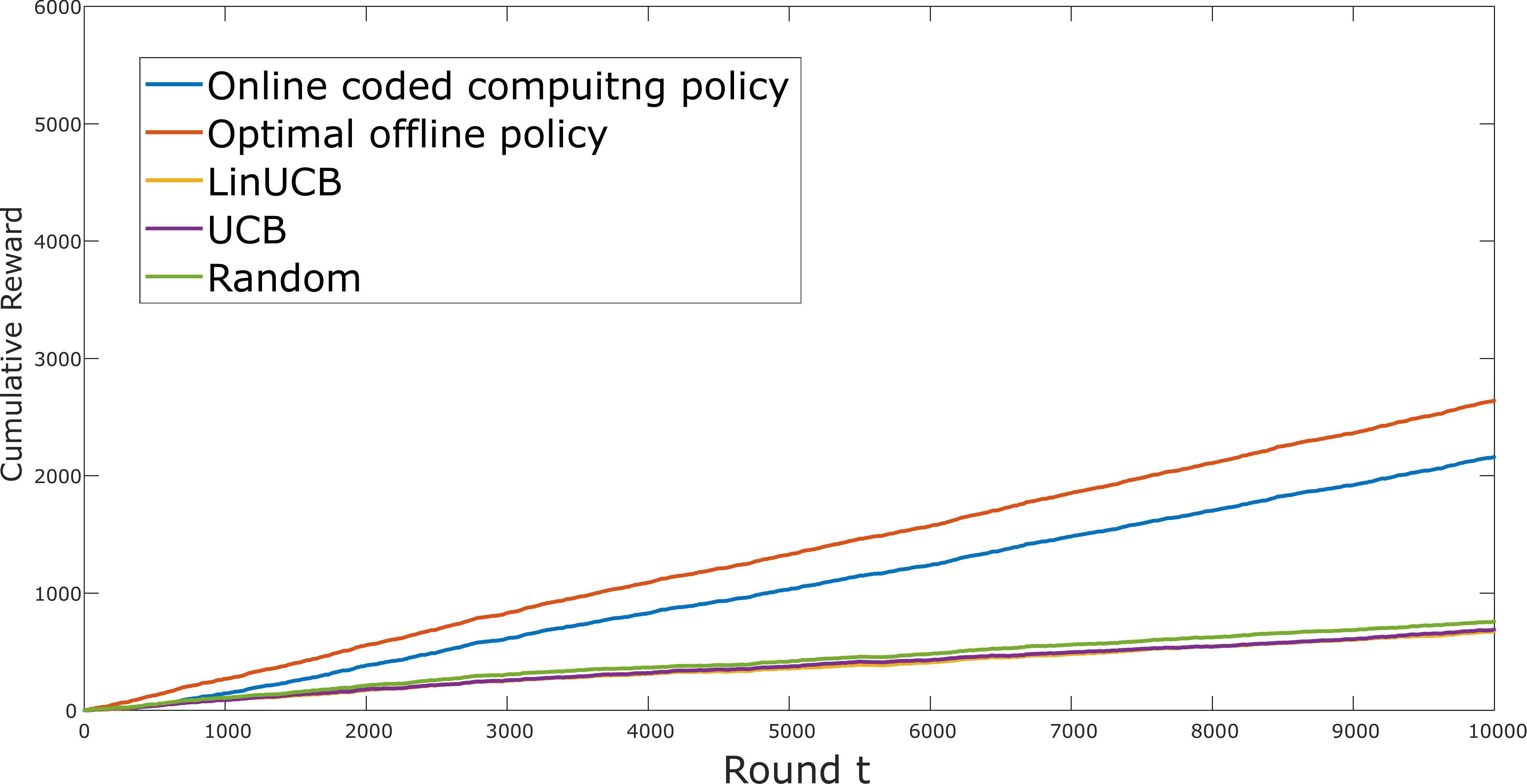}
\caption{Numerical evaluations for cumulative reward for \text{Scenario 1}.}
\label{fig:numerical1}
\end{figure}
\begin{figure}[t]
  \centering
    \includegraphics[width = \columnwidth]{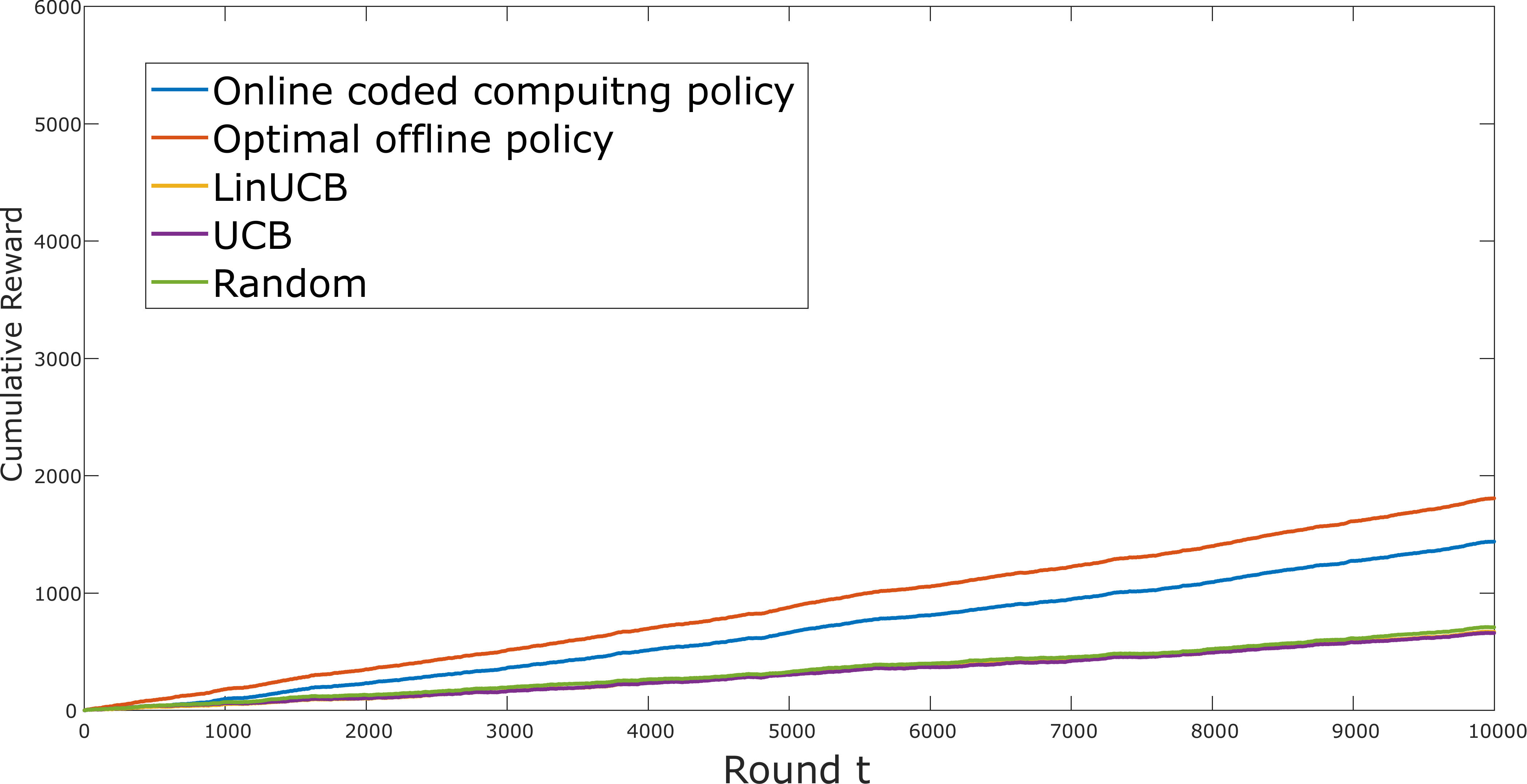}
\caption{Numerical evaluations for cumulative reward for \text{Scenario 2}.}
\label{fig:numerical2}
\end{figure}
For each scenario, the following benchmarks are considered to compare with the online coded edge computing policy:
\begin{enumerate}[leftmargin=*]
    \item \textbf{Optimal Offline policy}: Assuming knowledge of the success probability of each edge device in each round, the optimal set of edge devices is selected via Algorithm \ref{alg:optimal_oracle}.
\item \textbf{LinUCB~\cite{li2010contextual}}: LinUCB is a contextual-aware bandit algorithm which picks one arm in each round. We obtain a set of edge devices by repeating $b^t$ times of LinUCB. By sequentially removing selected edge devices, we ensure that the $b^t$ chosen edge devices are distinct.
\item \textbf{UCB~\cite{auer2002finite}}: UCB algorithm is a non-contextual and non-combinatorial algorithm. Similar to LinUCB, we repeat UCB $b^t$ times to select edge devices.
\item \textbf{Random}: A set of edge devices with size of $b^t$ is selected randomly from the available edge devices in each round $t$.
\end{enumerate}
Fig. \ref{fig:numerical1} to Fig. \ref{fig:numerical4} provide the cumulative rewards comparison of the online coded edge computing policy with the other $4$ benchmarks. We make the following conclusions from Fig. \ref{fig:numerical1} to Fig. \ref{fig:numerical4}:
\begin{itemize}[leftmargin=*]
    \item The optimal offline policy achieves the highest reward which gives an upper bound to the other policies. After a period of exploration, the proposed online policy is able to exploit the learned knowledge, and the cumulative reward approaches the upper bound.
    \item The proposed online coded edge computing policy significantly outperforms other benchmarks by taking into account the context of edge computing network.
    \item Random and UCB algorithms are not effective since they do not take the context into account for the decisions. Although LinUCB is a contextual-aware algorithm, it achieves similar cumulative regret as random and UCB algorithms. That is because the success probability model is more general here than the linear functions that LinUCB is tailored for. 
\end{itemize}
Fig. \ref{fig:regret} presents the expected regret of the proposed policy for Scenario 1. We can conclude that the proposed policy achieves a sublinear regret in the time horizon $T$ demonstrates the asymptotic optimality, i.e., $\lim_{T \rightarrow \infty}\frac{R(T)}{T} = 0$.
\begin{figure}[t]
  \centering
    \includegraphics[width = \columnwidth]{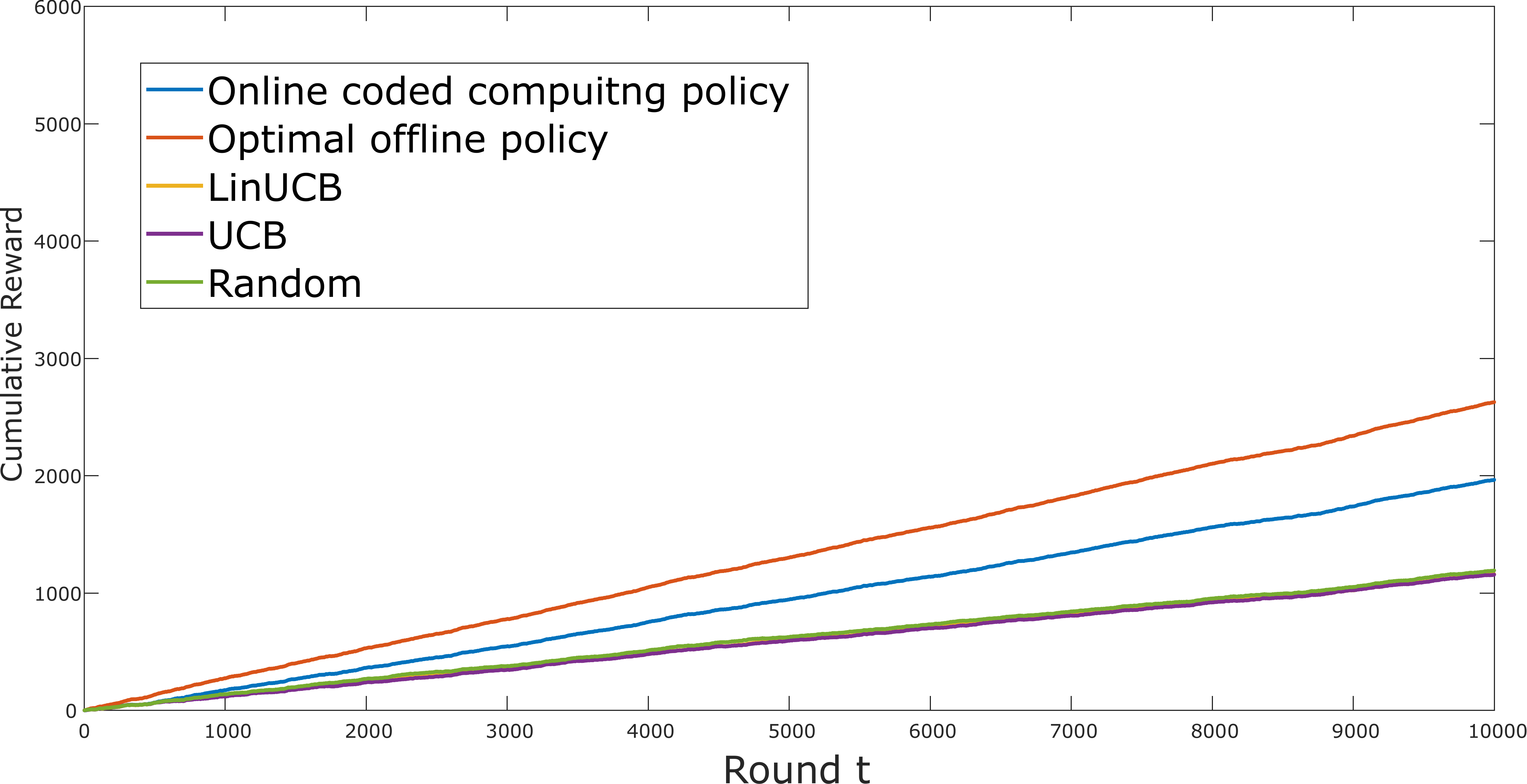}
\caption{Numerical evaluations for cumulative reward for \text{Scenario 3}.}
\label{fig:numerical3}
\end{figure}
\begin{figure}[t]
  \centering
    \includegraphics[width = \columnwidth]{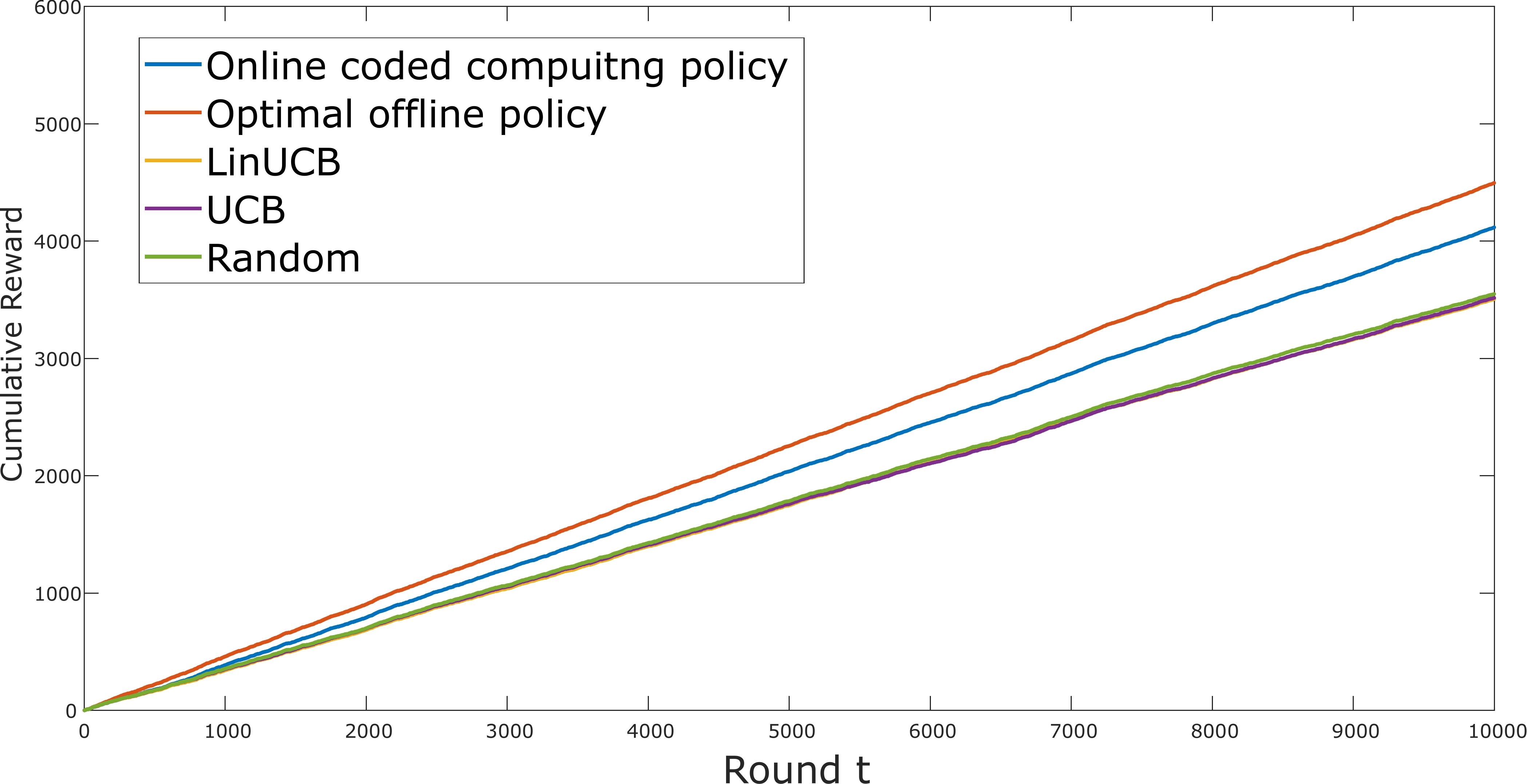}
\caption{Numerical evaluations for cumulative reward for \text{Scenario 4}.}
\label{fig:numerical4}
\end{figure}
\begin{figure}[t]
  \centering
    \includegraphics[width = \columnwidth]{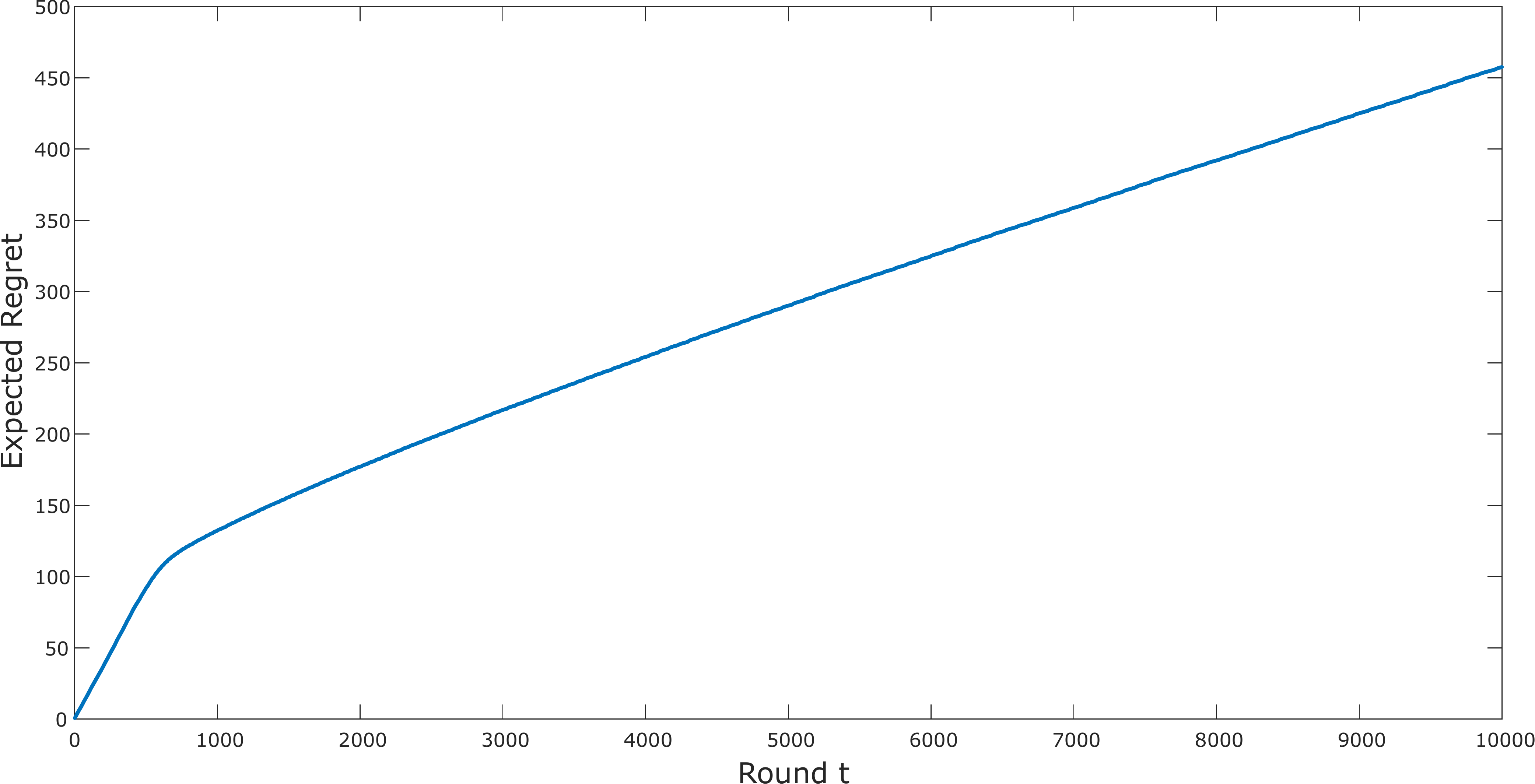}
\caption{Expected regret of the online coded edge computing policy for Scenario 1.}
\label{fig:regret}
\end{figure}

\section{Concluding Remarks and Future Directions}
Motivated by the volatility of edge devices' computing capabilities and the quality of service, and increasing demand for timely event-driven computations, we consider the problem of online computation offloading over unknown edge cloud networks without the knowledge of edge devices' capabilities. Under the coded computing framework, we formulate a combinatorial-contextual multiarmed bandit (CC-MAB) problem, which aims to maximize the cumulative expected reward. We propose the online coded edge computing policy which provably achieves asymptotically-optimal performance in terms of timely throughput, since the regret loss for the proposed CC-MAB problem compared with the optimal offline policy is sublinear. Finally, we show that the proposed online coded edge computing policy significantly improves the cumulative reward compared to the other benchmarks via numerical studies. 

%There are many interesting directions that can be pursued on the problem of online edge computing in the dark. First, the fundamental limit of achievable regret is not known, and developing a tight lower bound for regret remains to be resolved. Moreover, our proposed problem setting can be potentially extended to other computation models such as replications of jobs or redundant requests, offloading computations of jobs with precedence constraints that are modeled by directed acyclic graphs (DAGs), and other coded computing frameworks such as coded gradient aggregation~\cite{tandon2017gradient}.

\bibliographystyle{ieeetr}
\bibliography{references}

\begin{thebibliography}{10}

\bibitem{yang2020online}
C.-S. Yang, A.~S. Avestimehr, and R.~Pedarsani, ``Coded computing in unknown
  environment via online learning,'' in {\em 2020 IEEE International Symposium
  on Information Theory (ISIT)}.

\bibitem{zaharia2008improving}
M.~Zaharia, A.~Konwinski, A.~D. Joseph, R.~H. Katz, and I.~Stoica, ``Improving
  mapreduce performance in heterogeneous environments.,'' in {\em Osdi},
  vol.~8, p.~7, 2008.

\bibitem{ananthanarayanan2013effective}
G.~Ananthanarayanan, A.~Ghodsi, S.~Shenker, and I.~Stoica, ``Effective
  straggler mitigation: Attack of the clones.,'' in {\em NSDI}, vol.~13,
  pp.~185--198, 2013.

\bibitem{yu2019lagrange}
Q.~Yu, S.~Li, N.~Raviv, S.~M. Mousavi, M.~Soltanolkotabi, and A.~S. Avestimehr,
  ``Lagrange coded computing: Optimal design for resiliency, security and
  privacy,'' in {\em Artificial Intelligence and Statistics}, 2019.

\bibitem{eryilmaz2005stable}
A.~Eryilmaz, R.~Srikant, and J.~R. Perkins, ``Stable scheduling policies for
  fading wireless channels,'' {\em IEEE/ACM Transactions on Networking},
  vol.~13, no.~2, pp.~411--424, 2005.

\bibitem{tassiulas1992stability}
L.~Tassiulas and A.~Ephremides, ``Stability properties of constrained queueing
  systems and scheduling policies for maximum throughput in multihop radio
  networks,'' {\em IEEE transactions on automatic control}, vol.~37, no.~12,
  pp.~1936--1948, 1992.

\bibitem{dai2005maximum}
J.~G. Dai and W.~Lin, ``Maximum pressure policies in stochastic processing
  networks,'' {\em Operations Research}, vol.~53, no.~2, 2005.

\bibitem{neely2005dynamic}
M.~J. Neely, E.~Modiano, and C.~E. Rohrs, ``Dynamic power allocation and
  routing for time-varying wireless networks,'' {\em IEEE Journal on Selected
  Areas in Communications}, vol.~23, no.~1, pp.~89--103, 2005.

\bibitem{maguluri2012stochastic}
S.~T. Maguluri, R.~Srikant, and L.~Ying, ``Stochastic models of load balancing
  and scheduling in cloud computing clusters,'' in {\em INFOCOM, 2012
  Proceedings IEEE}, pp.~702--710, IEEE, 2012.

\bibitem{yang2019communication}
C.-S. Yang, R.~Pedarsani, and A.~S. Avestimehr, ``Communication-aware
  scheduling of serial tasks for dispersed computing,'' {\em IEEE/ACM
  Transactions on Networking (TON)}, vol.~27, no.~4, pp.~1330--1343, 2019.

\bibitem{hoseinnejhad2017deadline}
M.~Hoseinnejhad and N.~J. Navimipour, ``Deadline constrained task scheduling in
  the cloud computing using a discrete firefly algorithm,'' {\em INTERNATIONAL
  JOURNAL OF NEXT-GENERATION COMPUTING}, vol.~8, no.~3, 2017.

\bibitem{lee2018speeding}
K.~Lee, M.~Lam, R.~Pedarsani, D.~Papailiopoulos, and K.~Ramchandran, ``Speeding
  up distributed machine learning using codes,'' {\em IEEE Transactions on
  Information Theory}, vol.~64, no.~3, pp.~1514--1529, 2018.

\bibitem{li2018fundamental}
S.~Li, M.~A. Maddah-Ali, Q.~Yu, and A.~S. Avestimehr, ``A fundamental tradeoff
  between computation and communication in distributed computing,'' {\em IEEE
  Transactions on Information Theory}, vol.~64, no.~1, pp.~109--128, 2018.

\bibitem{dutta2016short}
S.~Dutta, V.~Cadambe, and P.~Grover, ``Short-dot: Computing large linear
  transforms distributedly using coded short dot products,'' in {\em Advances
  In Neural Information Processing Systems}, pp.~2100--2108, 2016.

\bibitem{lee2017high}
K.~Lee, C.~Suh, and K.~Ramchandran, ``High-dimensional coded matrix
  multiplication,'' in {\em Information Theory (ISIT), 2017 IEEE International
  Symposium on}, pp.~2418--2422, IEEE, 2017.

\bibitem{yu2017polynomial}
Q.~Yu, M.~Maddah-Ali, and S.~Avestimehr, ``Polynomial codes: an optimal design
  for high-dimensional coded matrix multiplication,'' in {\em Advances in
  Neural Information Processing Systems}, pp.~4403--4413, 2017.

\bibitem{tandon2017gradient}
R.~Tandon, Q.~Lei, A.~G. Dimakis, and N.~Karampatziakis, ``Gradient coding:
  Avoiding stragglers in distributed learning,'' in {\em International
  Conference on Machine Learning}, pp.~3368--3376, 2017.

\bibitem{li2017coding}
S.~Li, M.~A. Maddah-Ali, and A.~S. Avestimehr, ``Coding for distributed fog
  computing,'' {\em IEEE Communications Magazine}, vol.~55, no.~4, pp.~34--40,
  2017.

\bibitem{li2020coded}
S.~Li and S.~Avestimehr, ``Coded computing,'' {\em Foundations and
  Trends{\textregistered} in Communications and Information Theory}, vol.~17,
  no.~1, 2020.

\bibitem{prakash2020codedgraph}
S.~Prakash, A.~Reisizadeh, R.~Pedarsani, and A.~S. Avestimehr, ``Coded
  computing for distributed graph analytics,'' {\em IEEE Transactions on
  Information Theory}, vol.~66, no.~10, pp.~6534--6554, 2020.

\bibitem{yu2020straggler}
Q.~Yu, M.~Ali, and A.~S. Avestimehr, ``Straggler mitigation in distributed
  matrix multiplication: Fundamental limits and optimal coding,'' {\em IEEE
  Transactions on Information Theory}, 2020.

\bibitem{reisizadeh2019coded}
A.~Reisizadeh, S.~Prakash, R.~Pedarsani, and A.~S. Avestimehr, ``Coded
  computation over heterogeneous clusters,'' {\em IEEE Transactions on
  Information Theory}, 2019.

\bibitem{ferdinand2018hierarchical}
N.~Ferdinand and S.~C. Draper, ``Hierarchical coded computation,'' in {\em 2018
  IEEE International Symposium on Information Theory (ISIT)}, pp.~1620--1624,
  IEEE, 2018.

\bibitem{chen2018draco}
L.~Chen, H.~Wang, Z.~Charles, and D.~Papailiopoulos, ``Draco:
  Byzantine-resilient distributed training via redundant gradients,'' in {\em
  International Conference on Machine Learning}, pp.~903--912, 2018.

\bibitem{yang2021codedboolean}
C.-S. Yang and A.~S. Avestimehr, ``Coded computing for secure boolean
  computations,'' {\em IEEE Journal on Selected Areas in Information Theory},
  2021.

\bibitem{so2019codedprivateml}
J.~So, B.~Guler, A.~S. Avestimehr, and P.~Mohassel, ``Codedprivateml: A fast
  and privacy-preserving framework for distributed machine learning,'' {\em
  arXiv preprint arXiv:1902.00641}, 2019.

\bibitem{so2020scalable}
J.~So, B.~Guler, and A.~S. Avestimehr, ``A scalable approach for
  privacy-preserving collaborative machine learning,'' {\em arXiv preprint
  arXiv:2011.01963}, 2020.

\bibitem{karakus2017straggler}
C.~Karakus, Y.~Sun, S.~Diggavi, and W.~Yin, ``Straggler mitigation in
  distributed optimization through data encoding,'' in {\em Advances in Neural
  Information Processing Systems}, pp.~5434--5442, 2017.

\bibitem{so2020turbo}
J.~So, B.~Guler, and A.~S. Avestimehr, ``Turbo-aggregate: Breaking the
  quadratic aggregation barrier in secure federated learning,'' {\em IEEE
  Journal on Selected Areas in Information Theory}, 2021.

\bibitem{prakash2020coded}
S.~Prakash, S.~Dhakal, M.~R. Akdeniz, Y.~Yona, S.~Talwar, S.~Avestimehr, and
  N.~Himayat, ``Coded computing for low-latency federated learning over
  wireless edge networks,'' {\em IEEE Journal on Selected Areas in
  Communications}, vol.~39, no.~1, pp.~233--250, 2020.

\bibitem{prakash2020hierarchical}
S.~Prakash, A.~Reisizadeh, R.~Pedarsani, and A.~S. Avestimehr, ``Hierarchical
  coded gradient aggregation for learning at the edge,'' in {\em 2020 IEEE
  International Symposium on Information Theory (ISIT)}, pp.~2616--2621, IEEE,
  2020.

\bibitem{yu2020codedtree}
M.~Yu, S.~Sahraei, S.~Li, S.~Avestimehr, S.~Kannan, and P.~Viswanath, ``Coded
  merkle tree: Solving data availability attacks in blockchains,'' in {\em
  International Conference on Financial Cryptography and Data Security},
  pp.~114--134, Springer, 2020.

\bibitem{li2020polyshardTIFS}
S.~Li, M.~Yu, C.-S. Yang, A.~S. Avestimehr, S.~Kannan, and P.~Viswanath,
  ``Polyshard: Coded sharding achieves linearly scaling efficiency and security
  simultaneously,'' {\em IEEE Transactions on Information Forensics and
  Security}, vol.~16, pp.~249--261, 2020.

\bibitem{yang2019timely}
C.-S. Yang, R.~Pedarsani, and A.~S. Avestimehr, ``Timely-throughput optimal
  coded computing over cloud networks,'' in {\em Proceedings of the Twentieth
  ACM International Symposium on Mobile Ad Hoc Networking and Computing},
  pp.~301--310, ACM, 2019.

\bibitem{yang2019timelyisit}
C.-S. Yang, R.~Pedarsani, and A.~S. Avestimehr, ``Timely coded computing,'' in
  {\em 2019 IEEE International Symposium on Information Theory (ISIT)},
  pp.~2798--2802, IEEE, 2019.

\bibitem{lai1985asymptotically}
T.~L. Lai and H.~Robbins, ``Asymptotically efficient adaptive allocation
  rules,'' {\em Advances in applied mathematics}, vol.~6, no.~1, pp.~4--22,
  1985.

\bibitem{auer2002finite}
P.~Auer, N.~Cesa-Bianchi, and P.~Fischer, ``Finite-time analysis of the
  multiarmed bandit problem,'' {\em Machine learning}, vol.~47, no.~2-3,
  pp.~235--256, 2002.

\bibitem{li2010contextual}
L.~Li, W.~Chu, J.~Langford, and R.~E. Schapire, ``A contextual-bandit approach
  to personalized news article recommendation,'' in {\em Proceedings of the
  19th international conference on World wide web}, pp.~661--670, ACM, 2010.

\bibitem{sen2017contextual}
R.~Sen, K.~Shanmugam, M.~Kocaoglu, A.~Dimakis, and S.~Shakkottai, ``Contextual
  bandits with latent confounders: An nmf approach,'' in {\em Artificial
  Intelligence and Statistics}, pp.~518--527, 2017.

\bibitem{shariff2018differentially}
R.~Shariff and O.~Sheffet, ``Differentially private contextual linear
  bandits,'' in {\em Advances in Neural Information Processing Systems},
  pp.~4296--4306, 2018.

\bibitem{gai2012combinatorial}
Y.~Gai, B.~Krishnamachari, and R.~Jain, ``Combinatorial network optimization
  with unknown variables: Multi-armed bandits with linear rewards and
  individual observations,'' {\em IEEE/ACM Transactions on Networking (TON)},
  vol.~20, no.~5, pp.~1466--1478, 2012.

\bibitem{li2019combinatorial}
F.~Li, J.~Liu, and B.~Ji, ``Combinatorial sleeping bandits with fairness
  constraints,'' in {\em IEEE INFOCOM 2019-IEEE Conference on Computer
  Communications}, pp.~1702--1710, IEEE, 2019.

\bibitem{chen2019task}
L.~Chen and J.~Xu, ``Task replication for vehicular cloud: Contextual
  combinatorial bandit with delayed feedback,'' in {\em IEEE INFOCOM 2019-IEEE
  Conference on Computer Communications}, pp.~748--756, IEEE, 2019.

\bibitem{li2016contextual}
S.~Li, B.~Wang, S.~Zhang, and W.~Chen, ``Contextual combinatorial cascading
  bandits,'' in {\em Proceedings of the 33rd International Conference on
  International Conference on Machine Learning-Volume 48}, pp.~1245--1253,
  2016.

\bibitem{muller2016context}
S.~M{\"u}ller, O.~Atan, M.~van~der Schaar, and A.~Klein, ``Context-aware
  proactive content caching with service differentiation in wireless
  networks,'' {\em IEEE Transactions on Wireless Communications}, vol.~16,
  no.~2, pp.~1024--1036, 2016.

\bibitem{qin2014contextual}
L.~Qin, S.~Chen, and X.~Zhu, ``Contextual combinatorial bandit and its
  application on diversified online recommendation,'' in {\em Proceedings of
  the 2014 SIAM International Conference on Data Mining}, pp.~461--469, SIAM,
  2014.

\bibitem{hoeffding1994probability}
W.~Hoeffding, ``Probability inequalities for sums of bounded random
  variables,'' in {\em The Collected Works of Wassily Hoeffding}, pp.~409--426,
  Springer, 1994.

\end{thebibliography}
\vspace{-3mm}
\begin{IEEEbiographynophoto}{Chien-Sheng Yang}
received his the B.S. degree in electrical and computer engineering from
National Chiao Tung University (NCTU), Hsinchu, Taiwan in 2015 and is currently pursuing his Ph.D. in Electrical and Computer Engineering from the University of Southern California (USC), Los Angeles. He received the Annenberg Graduate Fellowship in 2016. He was a finalist of the ACM International Symposium on Mobile Ad Hoc Networking and Computing (MobiHoc) Best Paper Award in 2019. His interests include information theory, machine learning and edge computing. 
\end{IEEEbiographynophoto}
\vspace{-3mm}
\begin{IEEEbiographynophoto}{Ramtin Pedarsani}
received the B.Sc. degree in electrical engineering from the University of Tehran, Tehran, Iran, in 2009, the M.Sc. degree in communication systems from the Swiss Federal Institute of Technology (EPFL), Lausanne, Switzerland, in 2011, and the Ph.D. degree in electrical engineering and computer sciences from the University of California at Berkeley, Berkeley, CA, USA, in 2015.,He is currently an Assistant Professor with the Department of Electrical and Computer Engineering, University of California, Santa Barbara, CA, USA. His research interests include machine learning, optimization, information theory, game theory, and transportation systems., Dr. Pedarsani is the recipient of the Communications Society and Information Theory Society Joint Paper Award in 2020, the Best Paper Award at the IEEE International Conference on Communications in 2014, and the NSF CRII Award in 2017.
\end{IEEEbiographynophoto}
\begin{IEEEbiographynophoto}{A. Salman Avestimehr}
is a Professor, the inaugural director of the USC-Amazon Center on Secure and Trusted Machine Learning (Trusted AI), and the director of the Information Theory and Machine Learning (vITAL) research lab at the Electrical and Computer Engineering Department of University of Southern California. He is also an Amazon Scholar at Alexa AI. He received his Ph.D. in 2008 and M.S. degree in 2005 in Electrical Engineering and Computer Science, both from the University of California, Berkeley. Prior to that, he obtained his B.S. in Electrical Engineering from Sharif University of Technology in 2003. His research interests include information theory and coding theory, and large-scale distributed computing and machine learning, secure and private computing, and blockchain systems

Dr. Avestimehr has received a number of awards for his research, including the James L. Massey Research $\&$ Teaching Award from IEEE Information Theory Society, an Information Theory Society and Communication Society Joint Paper Award, a Presidential Early Career Award for Scientists and Engineers (PECASE) from the White House, a Young Investigator Program (YIP) award from the U.S. Air Force Office of Scientific Research, a National Science Foundation CAREER award, the David J. Sakrison Memorial Prize, and several Best Paper Awards at Conferences. He has been an Associate Editor for IEEE Transactions on Information Theory. He is currently a general Co-Chair of the 2020 International Symposium on Information Theory (ISIT).
\end{IEEEbiographynophoto}
\appendices
\section{Proof of Lemma \ref{lemma:R_e}} \label{proof_lemma_Re}
Suppose the policy enters the exploration phase in round $t$ and let $\mathcal{P}^t = \{p^t_{\nu}\}_{\nu \in \mathcal{V}^t}$ be the corresponding hypercubes of edge devices. Then, based on the design of the proposed policy, the set of under-explored hypercubes $\mathcal{P}^{\text{ue},t}_T$ is non-empty, i.e., there exists at least one edge device with context $\phi^t_{\nu}$ such that a hypercube $p$ satisfying $\phi^t_{\nu} \in p$ has $C^t(p) \leq K(t)= t^z \log{(t)}$. Clearly, there can be at most $\lceil T^z \log{(T) \rceil}$ exploration phases in which edge devices with contexts in $p$ are selected due to under-exploration of $p$. Since there are $(h_T)^D$ hypercubes in the partition, there can be at most $(h_T)^DT^z\log{(T)}$ exploration phases. Also, the maximum achievable
reward of an offloading decision is bounded by $1 - \eta$ and the minimum achievable reward
is $-B\eta$. The maximum regret in one exploration phase is bounded by $1+\eta(B - 1) < 1 + \eta B$. Therefore, we have
\begin{align}
    \mathbb{E}[R_e(T)] & \leq (1+\eta B) (h_T)^D \lceil T^z \log{(T)}\rceil \\
    & = (1+\eta B) \lceil T^{\gamma} \rceil^D \lceil T^z \log{(T)}\rceil \\
    & \leq (1+\eta B) 2^DT^{\gamma D} (T^z \log{(T)}+1) \\
    & = (1+\eta B) 2^D (T^{z+\gamma D}\log{(T)}+T^{\gamma D}) 
\end{align}
using the fact that $\lceil T^{\gamma} \rceil^D \leq (2T^{\gamma})^D = 2^D T^{\gamma D}$.

\section{Proof of Lemma \ref{lemma:R_s}} \label{proof_lemma_Rs}
For each $t \in [T]$, we define $W^t = \{\mathcal{P}^{\text{ue},t} = \emptyset \}$ as the event that the algorithm enters the exploitation phase. By the definition of $\mathcal{P}^{\text{ue},t}$, we have that $C^t(p)>K(t) = t^z\log{(t)}$ for all $p \in \mathcal{P}^t$. Let $V^t_{G}$ be the event that subset $G \in \mathcal{L}^t$ is selected in round $t$. Then, we have
\begin{align}
    R_s(T) = \sum^T_{t=1} \sum_{G \in \mathcal{L}^t} \mathbbm{1}_{\{V^t_G,W^t\}} \times (r(\mathcal{A}^{t*}) - r(G)).
\end{align}
Since the maximum regret is bounded by $1+\eta B$, we have
\begin{align}
     R_s(T) \leq (1+\eta B)\sum^T_{t=1} \sum_{G \in \mathcal{L}^t} \mathbbm{1}_{\{V^t_G,W^t\}}.
\end{align}
By taking the expectation, the regret can be bounded as follows
\begin{align}
    \mathbb{E}[R_s(T)] \leq  (1+\eta B)\sum^T_{t=1} \sum_{G \in \mathcal{L}^t} \textrm{Pr}(V^t_G,W^t). 
\end{align}
Now, we explain how to bound $\textrm{Pr}(V^t_G,W^t)$. Because of the design of policy, the choice of $G$ is optimal based on the estimated $\hat{\boldsymbol{\mu}}^t$. Thus, we have $u(\hat{\boldsymbol{\mu}}^t,G) \geq u(\hat{\boldsymbol{\mu}}^t,\tilde{\mathcal{A}}^t)$ which implies
\begin{align}
    \textrm{Pr}(V^t_G,W^t) \leq \textrm{Pr}(u(\hat{\boldsymbol{\mu}}^t,G) \geq u(\hat{\boldsymbol{\mu}}^t,\tilde{\mathcal{A}}^t),W^t).
\end{align}
The event $\{u(\hat{\boldsymbol{\mu}}^t,G) \geq u(\hat{\boldsymbol{\mu}}^t,\tilde{\mathcal{A}}^t),W^t\}$ actually implies that at least one of the following events holds for any $H(t) > 0$:
\begin{align}
    E_1 = \big \{u(\hat{\boldsymbol{\mu}}^t,G) \geq u(\overline{\boldsymbol{\mu}}^t,G)+H(t),W^t \big\} \nonumber\\
    E_2 = \big\{u(\hat{\boldsymbol{\mu}}^t,\tilde{\mathcal{A}}^t) \leq u(\underline{\boldsymbol{\mu}}^t,\tilde{\mathcal{A}}^t)-H(t),W^t \big\} \nonumber\\
    E_3 = \big\{ 
    u(\hat{\boldsymbol{\mu}}^t,G) \geq u(\hat{\boldsymbol{\mu}}^t,\tilde{\mathcal{A}}^t),
    u(\hat{\boldsymbol{\mu}}^t,G) < u(\overline{\boldsymbol{\mu}}^t,G)+H(t), \nonumber \\  u(\hat{\boldsymbol{\mu}}^t,\tilde{\mathcal{A}}^t) > u(\underline{\boldsymbol{\mu}}^t,\tilde{\mathcal{A}}^t)-H(t),W^t\big\}. \nonumber
\end{align}
Therefore, we have $\big \{u(\hat{\boldsymbol{\mu}}^t,G) \geq u(\hat{\boldsymbol{\mu}}^t,\tilde{\mathcal{A}}^t),W^t \big \} \subseteq E_1 \cup E_2 \cup E_3$. 

Then, we proceed to bound the probabilities of events $E_1$, $E_2$ and $E_3$ separately. Before bounding $\textrm{Pr}(E_1)$, we first present the following lemma which is proved in Appendix \ref{proof_lemma_mu12}.
\begin{lemma} \label{lemma:mu12}
Given a positive number $H(t)$, $\boldsymbol{\mu}_1$, $\boldsymbol{\mu}_2$ and $G$, if $u(\boldsymbol{\mu}_1,G) \geq u(\boldsymbol{\mu}_2,G)+H(t)$, then there exits $\nu \in G$ such that
\begin{align}
    \mu_1(p^t_{\nu}) \geq \mu_2(p^t_{\nu})+\frac{H(t)}{BM},
\end{align}
 where $B = \max_{1 \leq t \leq T}b^t$ and $M = \max_{1\leq t \leq T}\binom{B-1}{Y^t-1}$.
\end{lemma}
Thus, by Lemma \ref{lemma:mu12}, we have $E_1 = \big \{u(\hat{\boldsymbol{\mu}}^t,G) \geq u(\overline{\boldsymbol{\mu}}^t,G)+H(t),W^t \big\} \subseteq \big\{\hat{\mu}^t(p^t_{\nu}) \geq \overline{\mu}(p^t_{\nu})+\frac{H(t)}{BM}, \exists \nu \in G,W^t \big\}$.

By the definition of $\overline{\mu}(p)$, the expectation of estimated success probability for the edge device $\nu \in \mathcal{V}^t$ can be bounded by $\mathbb{E}[\hat{\mu}^t(p^t_{\nu})] \leq \overline{\mu}(p^t_{\nu})$. Then, we bound $\textrm{Pr}(E_1)$ as follows
\begin{align}
    \textrm{Pr}(E_1) = \textrm{Pr}(u(\hat{\boldsymbol{\mu}}^t,G) \geq u(\overline{\boldsymbol{\mu}}^t,G)+H(t),W^t)\\
    \leq \textrm{Pr}(\hat{\mu}(p^t_{\nu}) \geq \overline{\mu}(p^t_{\nu})+\frac{H(t)}{BM}, \exists \nu \in G,W^t)\\
    \leq  \textrm{Pr}(\hat{\mu}^t(p^t_{\nu}) \geq \mathbb{E}[\hat{\mu}^t(p^t_{\nu})]+\frac{H(t)}{BM}, \exists \nu \in G,W^t)\\
    \leq \sum_{\nu \in G} \textrm{Pr}(\hat{\mu}^t(p^t_{\nu}) \geq \mathbb{E}[\hat{\mu}^t(p^t_{\nu})]+\frac{H(t)}{BM},W^t).
\end{align}
By applying Chernoff-Hoeffding inequality \cite{hoeffding1994probability} and the fact that there are at least $K(t) = t^z\log{(t)}$ samples drawn, we have
\begin{align}
    \textrm{Pr}(E_1) & \leq \sum_{\nu \in G} \textrm{Pr}(\hat{\mu}^t(p^t_{\nu}) \geq \mathbb{E}[\hat{\mu}^t(p^t_{\nu})]+\frac{H(t)}{BM},W^t)\\
    & \leq \sum_{\nu \in G} \exp{(\frac{-2C^t(p^t_{\nu})H(t)^2}{B^2M^2})}\\
    & \leq \sum_{\nu \in G} \exp{(\frac{-2t^z\log{(t)}H(t)^2}{B^2M^2})}. 
\end{align}
%Similarly, we have a bound for $\textrm{Pr}(E_2)$ as follows:
%\begin{align}
 %   \textrm{Pr}(E_2)  \leq \sum_{\nu \in G} \exp{(\frac{-2t^z\log{(t)}H(t)^2}{B^2M^2})}.
%\end{align}
If we choose $H(t) = BMt^{-z/2} > 0$, we have
\begin{align}
    \textrm{Pr}(E_1) & \leq B \exp{(\frac{-2t^z\log{(t)H(t)^2}}{B^2M^2})}\\
    & = B \exp{(-2\log{(t)})} = Bt^{-2}. \label{eq:bound_E1}
\end{align}
Similarly, we have a bound for $\textrm{Pr}(E_2)$:
\begin{align}
      \textrm{Pr}(E_2) & \leq Bt^{-2}. \label{eq:bound_E2}
\end{align}
Lastly, we bound $\textrm{Pr}(E_3)$. Now we suppose that the following condition is satisfied:
\begin{align} \label{eq:condition}
    2H(t) \leq At^{\theta}.
\end{align}
Since $G \in \mathcal{L}^t$, we have $u(\underline{\boldsymbol{\mu}}^t,\tilde{\mathcal{A}}^t)- u(\overline{\boldsymbol{\mu}}^t,G) \geq At^{\theta}$. With (\ref{eq:condition}), we have $u(\underline{\boldsymbol{\mu}}^t,\tilde{\mathcal{A}}^t)- H(t) \geq u(\overline{\boldsymbol{\mu}}^t,G) + H(t)$ which contradicts event $E_3$. That is, under condition (\ref{eq:condition}), we have $\textrm{Pr}(E_3) = 0$. 

Under condition (\ref{eq:condition}), using (\ref{eq:bound_E1}) and (\ref{eq:bound_E2}), we have %$\textrm{Pr}(V^t_G,W^t)$ can be bounded as follows:
\begin{align}
    \textrm{Pr}(V^t_G,W^t) & \leq \textrm{Pr}(E_1 \cup E_2 \cup E_3)\\ 
    & \leq \textrm{Pr}(E_1) +\textrm{Pr}(E_2) +\textrm{Pr}(E_3) \leq 2Bt^{-2}.
\end{align}
Finally, we complete the regret bound for $\mathbb{E}[R_s(T)]$ as follows:
\begin{align}
    \mathbb{E}[&R_s(T)] \leq  (1+\eta B)\sum^T_{t=1} \sum_{G \in \mathcal{L}^t} \textrm{Pr}(V^t_G,W^t) \\
    & \leq (1+\eta B) |\mathcal{L}^t| \sum^T_{t=1} 2Bt^{-2} \leq (1+\eta B) |\mathcal{L}^t|(2B)\sum^{\infty}_{t=1} t^{-2}\\
    & = (1+\eta B) |\mathcal{L}^t| (2B) \frac{\pi^2}{6} \leq (1+\eta B) B \frac{\pi^2}{3} \sum^B_{k=1} \binom{|\mathcal{V}|}{k}.
\end{align}

\section{Proof of Lemma \ref{lemma:R_n}} \label{proof_lemma_Rn}

For each $t \in [T]$, we define $W^t = \{\mathcal{P}^{\text{ue},t} = \emptyset \}$ as the event that the policy enters the exploitation phase. Then, the regret due to near-optimal subsets can be written as
\begin{align}
    R_n(T) = \sum^T_{t=1} \mathbbm{1}_{\{W^t, G^t \in \mathcal{A}^t_{b^{-}}\backslash\mathcal{L}^t \}}(r(\mathcal{A}^{t*})-r(G^t)).
\end{align}
Let $Q^t = \{W^t,G^t \in \mathcal{A}^t_{b^{-}}\backslash\mathcal{L}^t\}$ be the event that a near-optimal subset is selected in round $t$. Then, we have
\begin{align}
    \mathbb{E}[R_n(T)] & = \sum^T_{t=1}\textrm{Pr}(Q^t)\mathbb{E}[r(\mathcal{A}^{t*})-r(G^t)|Q^t]\\
    & \leq \sum^T_{t=1} (u(\boldsymbol{\mu}^t\mathcal{A}^{t*})-u(\boldsymbol{\mu}^t,G^t)).
\end{align}
where $G^t$ is near-optimal in each round $t$. 
By the definition of $\mathcal{L}^t$, we then have
\begin{align}
    u(\underline{\boldsymbol{\mu}}^t,\tilde{\mathcal{A}}^t) - u(\overline{\boldsymbol{\mu}}^t,G^t) < At^{\theta}. \label{eq:def_L}
\end{align}
By the function $c$ defined in Appendix \ref{proof_lemma_mu12} and Assumption \ref{ass:holder}, we have
\begin{align}
    &u(\boldsymbol{\mu}^t,\mathcal{A}^{t*}) - u(\tilde{\boldsymbol{\mu}}^t,\mathcal{A}^{t*})=  c(\boldsymbol{\mu}^t,\tilde{\boldsymbol{\mu}}^t,\mathcal{A}^{t*},Y^t)\\
     \leq &\sum_{(G_1,G_2,\nu) \in \mathcal{S}(\mathcal{A}^{t*},Y^t)}|\mu(\phi^t_{\nu}) - \mu(\tilde{\phi}_{p^t_{\nu}})| \\
     \leq & \sum_{(G_1,G_2,\nu) \in \mathcal{S}(\mathcal{A}^{t*},Y^t)}
    L \| \phi^t_{\nu} - \tilde{\phi}_{p^t_{\nu}}\|^{\alpha} \\
     \leq & \sum_{(G_1,G_2,\nu) \in \mathcal{S}(\mathcal{A}^{t*},Y^t)} LD^{\frac{\alpha}{2}}h^{-\alpha}_T \\
     = & \binom{|\mathcal{A}^{t*}|}{Y^t}Y^t LD^{\frac{\alpha}{2}}h^{-\alpha}_T =  \binom{|\mathcal{A}^{t*}|-1}{Y^t-1}|\mathcal{A}^{t*}| LD^{\frac{\alpha}{2}}h^{-\alpha}_T\\
     \leq & BMLD^{\frac{\alpha}{2}}h^{-\alpha}_T.
\end{align}
Similarly, we have the following inequalities:
\begin{align}
    u(\tilde{\boldsymbol{\mu}}^t,\tilde{\mathcal{A}}^{t}) - u(\underline{\boldsymbol{\mu}}^t,\tilde{\mathcal{A}}^{t}) \leq BMLD^{\frac{\alpha}{2}}h^{-\alpha}_T \\
    u(\overline{\boldsymbol{\mu}}^t,G^{t}) - u(\boldsymbol{\mu}^t,G^t) \leq BMLD^{\frac{\alpha}{2}}h^{-\alpha}_T 
\end{align}
Now, we bound $u(\boldsymbol{\mu}^t\mathcal{A}^{t*})-u(\boldsymbol{\mu}^t,G^t)$ as follows:
\begin{align}
     & u(\boldsymbol{\mu}^t,\mathcal{A}^{t*})-u(\boldsymbol{\mu}^t,G^t)\\
     \leq &  u(\tilde{\boldsymbol{\mu}}^t,\mathcal{A}^{t*})+BMLD^{\frac{\alpha}{2}}h^{-\alpha}_T-u(\boldsymbol{\mu}^t,G^t)\\
          \leq &  u(\tilde{\boldsymbol{\mu}}^t,\tilde{\mathcal{A}}^{t})+BMLD^{\frac{\alpha}{2}}h^{-\alpha}_T-u(\boldsymbol{\mu}^t,G^t)\\
    \leq &  u(\underline{\boldsymbol{\mu}}^t,\tilde{\mathcal{A}}^{t})+2BMLD^{\frac{\alpha}{2}}h^{-\alpha}_T-u(\boldsymbol{\mu}^t,G^t)\\
    \leq &  u(\underline{\boldsymbol{\mu}}^t,\tilde{\mathcal{A}}^{t})+3BMLD^{\frac{\alpha}{2}}h^{-\alpha}_T-u(\overline{\boldsymbol{\mu}}^t,G^t)\\
    \leq & 3BMLD^{\frac{\alpha}{2}}h^{-\alpha}_T+At^{\theta}
\end{align}
by the definition of $\tilde{\mathcal{A}}^t$ and (\ref{eq:def_L}).
With $h_T = \lceil T^{\gamma}\rceil$, we have
\begin{align}
    u(\boldsymbol{\mu}^t\mathcal{A}^{t*})-u(\boldsymbol{\mu}^t,G^t) & \leq 3BMLD^{\frac{\alpha}{2}} \lceil T^{\gamma} \rceil^{-\alpha}+At^{\theta} \\
    & \leq 3BMLD^{\frac{\alpha}{2}} T^{- \alpha \gamma}+At^{\theta}.
\end{align}
Thus, we complete the regret bound for $\mathbb{E}[R_n]$ as follows:
\begin{align}
    \mathbb{E}[R_n] & \leq \sum^T_{t=1} (3BMLD^{\frac{\alpha}{2}} T^{- \alpha \gamma}+At^{\theta})\\
    & \leq 3BMLD^{\frac{\alpha}{2}} T^{1- \alpha \gamma} + \frac{A}{1+\theta} T^{1+\theta}.
\end{align}

\section{Proof of Lemma \ref{lemma:mu12}} \label{proof_lemma_mu12}
First, we suppose that 
\begin{align}
    \mu_1(p^t_{\nu}) - \mu_2(p^t_{\nu}) < \frac{H(t)}{BM}, \ \forall \nu \in G. \label{eq:assump}
\end{align}
We note that the following equation holds and will be used for analysis later.
\begin{align}
    \prod^N_{i=1}a_i - \prod^N_{i=1}b_i = \sum^N_{i=1}a_1\dots a_{i-1}(a_i-b_i)b_{i+1}\dots b_N. \label{eq:series}
\end{align}
Without loss of generality, we can index the elements in $G$ by $G = \{1,2,3\dots,|G|\}$. Then we define a function $c(\boldsymbol{\mu}_1,\boldsymbol{\mu}_2,G,Y) \triangleq u(\boldsymbol{\mu}_1,G)-u(\boldsymbol{\mu}_2,G)$, i.e.,
\begin{align}
    & c(\boldsymbol{\mu}_1, \boldsymbol{\mu}_2,G,Y^t) \nonumber\\ = &\sum^{|G|}_{s = Y^t} \sum_{G' \subseteq G, |G'| =s} \big \{\prod_{\nu \in G'} \mu_1 (p^t_{\nu}) \prod_{\nu \in G \backslash G'} (1-\mu_1(p^t_{\nu})) \nonumber\\
    & - \prod_{\nu \in G'} \mu_2 (p^t_{\nu}) \prod_{\nu \in G \backslash G'} (1-\mu_2(p^t_{\nu})\big \}.
\end{align}
We first define a function $f(G_1,G_2,\nu)$ as follows
\begin{align}
    f(G_1,G_2,\nu) \triangleq \prod_{\nu_1 \in G_1,\nu_1 < \nu}(1-\mu_1(p^t_{\nu_1})) \prod_{\nu_1 \in G_2,\nu_1 < \nu}\mu_1(p^t_{\nu_1}) \times \nonumber \\\prod_{\nu_2 \in G_1,\nu_2 > \nu}(1-\mu_2(p^t_{\nu_2}))\prod_{\nu_2 \in G_2,\nu_2 > \nu}\mu_1(p^t_{\nu_2}) \{\mu_1(p^t_{\nu})-\mu_2(p^t_{\nu})\}; \nonumber
\end{align}
and a set $\mathcal{S}(G,Y^t) \triangleq \{ (G_1,G_2,\nu):|G_1|=|G|-Y^t,|G_2|=Y^t,G_1\cup G_2 = G,\nu \in G_2\}$.

 We now show that $c(\boldsymbol{\mu}_1,\boldsymbol{\mu}_2,G,Y)$ can be rewritten as 
\begin{align}
  c(\boldsymbol{\mu}_1,\boldsymbol{\mu}_2,G,Y^t) =  \sum_{(G_1,G_2,\nu) \in \mathcal{S}(G,Y^t)}f(G_1,G_2 \backslash \{\nu\},\nu). \label{eq:c}
\end{align}
If $Y^t = |G|$, by equation (\ref{eq:series}), we have
\begin{align}
    & c(\boldsymbol{\mu}_1,\boldsymbol{\mu}_2,G,|G|) =  \prod_{\nu \in G} \mu_1 (p^t_{\nu}) -\prod_{\nu \in G} \mu_2 (p^t_{\nu}) \\
    = & \sum_{\nu \in G} \prod_{\nu_1 \in G, \nu_1 < \nu}\mu_1(p^t_{\nu_1})\prod_{\nu_2 \in G, \nu_2 > \nu}\mu_2(p^t_{\nu_2})\{\mu_1(p^t_{\nu})-\mu_2(p^t_{\nu})\} \\
    = & \sum_{(G_1,G_2,\nu) \in \mathcal{S}(G,|G|)} f(G_1, G_2 \backslash \{\nu\},\nu),
\end{align}
which implies that (\ref{eq:c}) holds for $Y^t = |G|$. 

Now we suppose that Equation (\ref{eq:c}) holds for $Y^t$, then we consider the case of $Y^t-1$. By the definition of function $c$, we have 
\begin{align}
      c(\boldsymbol{\mu}_1, &\boldsymbol{\mu}_2,G,Y^t-1) = c(\boldsymbol{\mu}_1, \boldsymbol{\mu}_2,G,Y^t) \nonumber\\
     &+ \sum_{G' \subseteq G, |G'| =Y^t-1} \big \{\prod_{\nu \in G'} \mu_1 (p^t_{\nu}) \prod_{\nu \in G \backslash G'} (1-\mu_1(p^t_{\nu}))\nonumber \\
     &- \prod_{\nu \in G'} \mu_2 (p^t_{\nu}) \prod_{\nu \in G \backslash G'} (1-\mu_2(p^t_{\nu})\big \}. \label{eq:Y-1}
\end{align}
Then, by using (\ref{eq:series}) and the definition of function $f$ and set $\mathcal{S}$, we can write the second term of (\ref{eq:Y-1}) as $\sum_{(G_1,G_2,\nu) \in \mathcal{S}(G,Y^t-1)}f(G_1,G_2 \backslash \{\nu\},\nu)- \sum_{(G_1,G_2,\nu) \in \mathcal{S}(G,|G|-Y^t+1)}f(G_2 \backslash \{\nu\},G_1,\nu)$.

For each $(G_1,G_2,\nu) \in \mathcal{S}(G,|G|-Y^t+1)$, the corresponding $(G_2 \backslash \{\nu\},G_1 \cup \{ \nu \}, \nu)$ is also in $\mathcal{S}(G,Y^t)$. Thus, we have
\begin{align}
    &\sum_{(G_1,G_2,\nu) \in \mathcal{S}(G,|G|-Y^t+1)}f(G_2 \backslash \{\nu\},G_1,\nu)\\  =  &\sum_{(G_1,G_2,\nu) \in \mathcal{S}(G,Y^t)}f(G_1,G_2 \backslash \{\nu\},\nu)= c(\boldsymbol{\mu}_1,\boldsymbol{\mu}_2,G,Y^t).
\end{align}
It follows that 
\begin{align}
    c(\boldsymbol{\mu}_1, &\boldsymbol{\mu}_2,G,Y^t-1) = \sum_{(G_1,G_2,\nu) \in \mathcal{S}(G,Y^t-1)}f(G_1,G_2 \backslash \{\nu\},\nu)
\end{align}
which implies that (\ref{eq:c}) holds for all $1\leq Y^t \leq |G|$.

With (\ref{eq:assump}) and the definition of function $f$, we have $f(G_1,G_2,\nu) < \frac{H(t)}{BM}$ for all $\nu$. Then we further have
\begin{align}
    & u(\boldsymbol{\mu}_1,G)-u(\boldsymbol{\mu}_2,G) \leq  \sum_{(G_1,G_2,\nu) \in \mathcal{S}(G,Y^t)} \frac{H(t)}{BM} \\
    = & \binom{|G|}{Y^t}Y^t\frac{H(t)}{BM}= \binom{|G|-1}{Y^t-1}|G|\frac{H(t)}{BM} \leq H(t)
\end{align}
which contradicts $u(\boldsymbol{\mu}_1,G) \geq u(\boldsymbol{\mu}_2,G)+H(t)$, i.e.,  there exits $\nu \in G$ such that $\mu_1(p^t_{\nu}) \geq \mu_2(p^t_{\nu})+\frac{H(t)}{BM}$.

\section{Proof of Lemma \ref{lemma:optimalset}} \label{proof_lemma_optimalset}
For a fixed integer $n_g$, we suppose $\mathcal{A}_1$ is the optimal set with cardinality $n_g$ where $i \notin \mathcal{A}_1$ and $1\leq i \leq n_g$. Thus, there exists a $j \in \mathcal{G}_1$ such that $j > n_g$. We construct a set $\mathcal{A}_2 = (\mathcal{A}_1 \backslash \{j\}) \cup \{i\}$, where $\mathcal{A}_1 \backslash \{j\} = \mathcal{A}_2 \backslash \{i\}$. Then, we have
$u(\boldsymbol{\mu}^t,\mathcal{A}_2) = \textrm{Pr}(\sum_{\nu \in \mathcal{A}_2} q^t_{\nu} \geq Y^t) - \eta n_g = \mu(p^t_i)\textrm{Pr}(\sum_{\nu \in \mathcal{A}_2\backslash \{i\}} q^t_{\nu} \geq Y^t-1) + (1-\mu(p^t_i))\textrm{Pr}(\sum_{\nu \in \mathcal{A}_2 \backslash \{i\}} q^t_{\nu} \geq Y^t) -\eta n_g$ and $u(\boldsymbol{\mu}^t,\mathcal{A}_1) = \mu(p^t_j)\textrm{Pr}(\sum_{\nu \in \mathcal{A}_1\backslash \{j\}} q^t_{\nu} \geq Y^t-1) + (1-\mu(p^t_j))\textrm{Pr}(\sum_{\nu \in \mathcal{A}_1 \backslash \{j\}} q^t_{\nu} \geq Y^t) - \eta n_g$. Then, we have $u(\boldsymbol{\mu}^t,\mathcal{A}_2) - u(\boldsymbol{\mu}^t,\mathcal{A}_1)=(\mu(p^t_i)-\mu(p^t_j))(\textrm{Pr}(\sum_{\nu \in \mathcal{A}_2\backslash \{i\}} q^t_{\nu} \geq Y^t-1)-\textrm{Pr}\{\sum_{\nu \in \mathcal{A}_2 \backslash \{i\}} q^t_{\nu} \geq Y^t)\}\geq  0$ which is a contradiction. 
\end{document}